%% file: main.tex
\def\anon{0} 
\def\conf{0} 
\def\papertitle{Direct Product Theorems for Randomized Query Complexity} 

\if\conf1
    \documentclass[conference]{IEEEtran}
\else
    \documentclass[11pt]{article}
    \pdfoutput=1 
\fi

\input{preamble}

\title{\papertitle}

\if\anon1
    \author{Anonymous submission}
\else
    \if\conf1
        \author{
            \IEEEauthorblockN{Shalev Ben{-}David}
            \IEEEauthorblockA{Institute for Quantum Computing \\
                              University of Waterloo \\
                              \texttt{shalev.b@uwaterloo.ca}}
            \and
            \IEEEauthorblockN{Eric Blais}
            \IEEEauthorblockA{University of Waterloo \\
                              \texttt{eric.blais@uwaterloo.ca}}
        }
    \else
        \author{
            Shalev Ben{-}David\\
            \small Institute for Quantum Computing\\
            \small University of Waterloo\\
            \small \texttt{shalev.b@uwaterloo.ca}
        \and
            Eric Blais\\
            \small University of Waterloo\\
            \small \texttt{eric.blais@uwaterloo.ca}
        }
    \fi
\fi

\date{}

\begin{document}

\maketitle

\begin{abstract}
We establish two new direct product theorems for the randomized query complexity of Boolean functions.

The first shows that computing $n$ copies of a function $f$, even with a small
success probability of $\gamma^n$, requires $\Theta(n)$ times the
\emph{maximum distributional} query complexity of $f$ with success parameter $\gamma$. This result holds for all success parameters $\gamma$, even when $\gamma$ is
very close to $1/2$ or to $1$.
As a result, it unifies and generalizes
Drucker's direct product theorem (2012) for $\gamma$ bounded away from $\frac12$ and $1$ as well as
the strong direct sum theorem of Blais and Brody (2019) for $\gamma\approx 1-1/n$.

The second establishes a general \emph{list decoding} direct product theorem that captures many different variants of ``partial computation'' tasks related to the function $f^n$ consisting of $n$ copies of $f$. Notably, our list decoding direct product theorem yields a new threshold direct product theorem and other new variants such as the labelled-threshold direct product theorem.

Both of these direct product theorems are obtained by taking a new approach. Instead of directly analyzing the query complexity of algorithms, we introduce a new measure of complexity of functions that we call \emph{discounted score}. We show that this measure satisfies a number of useful structural properties, including tensorization, that make it particularly suitable for the study of direct product questions.
\end{abstract}

\if\conf0
    \clearpage
    {\small\tableofcontents}
    \clearpage
\fi

\input{introduction}

\input{preliminaries}

\input{directproduct}

\input{tensorization}

\input{equivalence}

\input{listdecoding}

\input{boosting}

\input{relations}

\if\anon1
\else
\phantomsection\addcontentsline{toc}{section}{Acknowledgements}
\section*{Acknowledgements}
This research is supported in part by the Natural Sciences and
Engineering Research Council of Canada (NSERC), DGECR-2019-00027 and
RGPIN-2019-04804.\footnote{Cette recherche a été financée par le
Conseil de recherches en sciences naturelles et en génie du Canada
(CRSNG), DGECR-2019-00027 et RGPIN-2019-04804.}
\fi

\renewcommand{\UrlFont}{\ttfamily\small}
\let\oldpath\path
\renewcommand{\path}[1]{\small\oldpath{#1}}
\emergencystretch=2em 

\if\bibnonempty1{ 
\phantomsection\addcontentsline{toc}{section}{References} 
\printbibliography 
}\else{
\nocite{AA05} 
}
\fi

\end{document}

%% file: preamble.tex
\if\conf1
\else
  \usepackage[margin=1in]{geometry}
\fi

\usepackage{amsmath,amsthm,amsfonts,amssymb}
\usepackage{mathtools}
\usepackage{thmtools,thm-restate}

\usepackage{tikz}
\usetikzlibrary{calc}

\usepackage{silence} 
\usepackage[dvipsnames]{xcolor} 

\usepackage{mmap}
\usepackage[T1]{fontenc}

\usepackage[utf8]{inputenc}
\usepackage{csquotes}
\usepackage[english]{babel}

\usepackage[
    activate={true,nocompatibility},
    final, 
    tracking=true,
    kerning=true,
    spacing=true,
    factor=1100,
    stretch=12,
    shrink=10
    ]{microtype}
\microtypecontext{spacing=nonfrench}

\usepackage[backend=biber,
              style=alphabetic,
        maxbibnames=99,
      minalphanames=3,
      maxalphanames=4,
            backref=true]{biblatex}

\if\anon1
\usepackage{hyperref}
\hypersetup{
    pdftitle={\papertitle}, 
    pdfauthor={Anonymous submission}, 
    colorlinks=true, 
    linkcolor=NavyBlue, 
    citecolor=NavyBlue, 
    urlcolor=ForestGreen 
}
\else
\usepackage{hyperref}
\hypersetup{
    pdftitle={\papertitle}, 
    pdfauthor={Shalev Ben{-}David and Eric Blais}, 
    colorlinks=true, 
    linkcolor=blue, 
    citecolor=blue, 
    urlcolor=green 
}
\fi

\usepackage[capitalise, nameinlink]{cleveref}
\crefname{equation}{}{}
\usepackage{crossreftools}

\input{scripts.tex}

\newcommand\doverline[1]{%
\tikz[baseline=(nodeAnchor.base)]{
    \node[inner sep=0] (nodeAnchor) {$#1$}; 
    \draw[line width=0.1ex,line cap=round] 
        ($(nodeAnchor.north west)+(0.0em,0.2ex)$) -- ($(nodeAnchor.north east)+(0.0em,0.2ex)$) 
        ($(nodeAnchor.north west)+(0.0em,0.5ex)$) -- ($(nodeAnchor.north east)+(0.0em,0.5ex)$);
}}

\DeclareMathOperator*{\E}{\mathrm{E}}
\DeclareMathOperator*{\dom}{\mathrm{dom}}
\DeclareMathOperator*{\supp}{\mathrm{supp}}
\newcommand{\calR}{\mathcal{R}}

\newcommand{\cost}{\mathrm{cost}}
\newcommand{\barcost}{\overline{\cost}}
\newcommand{\scost}{\overline{\mathrm{scost}}}

\newcommand{\sss}{\mathrm{score}}
\newcommand{\bars}{\overline{\sss}}
\newcommand{\ssf}{\sss_f}
\newcommand{\sfmu}{\ssf^\mu}
\newcommand{\barsfmu}{\bars_f^\mu}

\newcommand{\success}{\mathrm{success}}
\newcommand{\barsuccess}{\overline{\success}}

\newcommand{\hel}{\mathrm{hel}}
\newcommand{\Hel}{\mathrm{Hel}}
\newcommand{\barHel}{\overline{\Hel}}
\newcommand{\HelLoss}{{\Hel_\downarrow}}
\newcommand{\barHelLoss}{\overline{\HelLoss}}
\newcommand{\ent}{\mathrm{ent}}
\newcommand{\Ent}{\mathrm{Ent}}
\newcommand{\barEnt}{\overline{\Ent}}

\newcommand{\ds}{\mathrm{ds}}
\newcommand{\maxDS}{\mathrm{DS}}

\newcommand{\D}{\mathrm{D}}

\newcommand{\avgD}{\doverline{\D}}
\newcommand{\R}{\mathrm{R}}
\newcommand{\barR}{\overline{\R}}
\newcommand{\avgR}{\doverline{\R}}
\newcommand{\sR}{\doverline{\mathrm{sR}}}

\newcommand{\LD}{\mathrm{LD}}
\newcommand{\Thr}{\mathrm{Thr}}
\newcommand{\Label}{\mathrm{Label}}

\newcommand{\calX}{\mathcal{X}}

\newtheorem{theorem}{Theorem}
\newtheorem{lemma}[theorem]{Lemma}
\newtheorem{claim}[theorem]{Claim}

\newtheorem{corollary}[theorem]{Corollary}
\newtheorem{proposition}[theorem]{Proposition}

\theoremstyle{definition}

\theoremstyle{plain}

%% file: scripts.tex
\AtEveryBibitem{
 \clearlist{address}
 \clearfield{date}
 \clearlist{location}
 \clearfield{month}
 \clearfield{day}
 \clearfield{series}
 \clearfield{pages}
 \clearlist{organization}
 \clearfield{number}
 \clearlist{language}

 \ifentrytype{book}{}{ 
  \clearlist{publisher}
  \clearname{editor}
  \clearfield{issn}
  \clearfield{isbn}
  \clearfield{volume}
 }
}

\DeclareFieldFormat{eprint:eccc}{
\mkbibacro{ECCC}\addcolon\space\ifhyperref
    {\href{http://eccc.hpi-web.de/report/#1/}{\nolinkurl{#1}}}
    {\nolinkurl{#1}}}
\DeclareFieldAlias{eprint:ECCC}{eprint:eccc}
\DeclareFieldFormat{eprint:iacr}{Cryptology ePrint\addcolon\space\ifhyperref
    {\href{https://eprint.iacr.org/#1}{\nolinkurl{#1}}}
    {\nolinkurl{#1}}}
\DeclareFieldAlias{eprint:IACR}{eprint:iacr}
\DeclareFieldFormat{eprint:arxiv}{arXiv\addcolon\space\ifhyperref
    {\href{https://arxiv.org/abs/#1}{\nolinkurl{#1}}}
    {\nolinkurl{#1}}}

\DeclareFieldFormat[article]{title}{#1}
\DeclareFieldFormat[inproceedings]{title}{#1}
\DeclareFieldFormat[incollection]{title}{#1}
\DeclareFieldFormat[online]{title}{#1}

\renewbibmacro{in:}{}

\newcommand{\lName}{1}

\newcommand{\donothing}[1]{#1}

\newcommand{\JACM}{\if\lName1\donothing{Journal of the {ACM}}\else{JACM}\fi}
\newcommand{\SICOMP}{\if\lName1\donothing{{SIAM} Journal on Computing}\else{SICOMP}\fi}
\newcommand{\ToC}{\if\lName1\donothing{Theory of Computing}\else{ToC}\fi}
\newcommand{\ToCGS}{\if\lName1\donothing{Theory of Computing Graduate Surveys}\else{ToC}\fi}
\newcommand{\TOCT}{\if\lName1\donothing{{ACM} Transactions on Computation Theory}\else{TOCT}\fi}
\newcommand{\ToIT}{\if\lName1\donothing{{IEEE} Transactions on Information Theory}\else{TOCT}\fi}
\newcommand{\CCjournal}{\if\lName1\donothing{Computational Complexity}\else{CC}\fi}
\newcommand{\CJTCS}{\if\lName1\donothing{Chicago Journal of Theoretical Computer Science}\else{CJTCS}\fi}

\newcommand{\TCS}{\if\lName1\donothing{Theoretical Computer Science}\else{TCS}\fi}
\newcommand{\IPL}{\if\lName1\donothing{Information Processing Letters}\else{IPL}\fi}
\newcommand{\JCSS}{\if\lName1\donothing{Journal of Computer and System Sciences}\else{JCSS}\fi}

\newcommand{\RSA}{\if\lName1\donothing{Random Structures and Algorithms}\else{RSA}\fi}
\newcommand{\JCTA}{\if\lName1\donothing{Journal of Combinatorial Theory, Series A}\else{JCTA}\fi}
\newcommand{\JCTB}{\if\lName1\donothing{Journal of Combinatorial Theory, Series B}\else{JCTB}\fi}
\newcommand{\PJM}{\if\lName1\donothing{Pacific Journal of Mathematics}\else{PJM}\fi}
\newcommand{\QICjournal}{\if\lName1\donothing{Quantum Information and Computation}\else{QIC}\fi}
\newcommand{\IJQI}{\if\lName1\donothing{International Journal of Quantum Information}\else{IJQI}\fi}
\newcommand{\PRA}{\if\lName1\donothing{Physical Review A}\else{PRA}\fi}
\newcommand{\PRL}{\if\lName1\donothing{Physical Review Letters}\else{PRL}\fi}
\newcommand{\VLDB}{\if\lName1\donothing{International Journal on Very Large Data Bases}\else{VLDB}\fi}
\newcommand{\SIDMA}{\if\lName1\donothing{{SIAM} Journal on Discrete Mathematics}\else{SIDMA}\fi}

\DefineBibliographyStrings{english}{%
  backrefpage = {p{.}},
  backrefpages = {pp{.}},
}

\def\bibnonempty{0}
\let\oldcite\cite
\renewcommand{\cite}[1]{\global\def\bibnonempty{1}\oldcite{#1}}

\newcommand{\eq}[1]{\hyperref[eq:#1]{(\ref*{eq:#1})}}
\renewcommand{\sec}[1]{\hyperref[sec:#1]{Section~\ref*{sec:#1}}}
\newcommand{\thm}[1]{\hyperref[thm:#1]{Theorem~\ref*{thm:#1}}}
\newcommand{\lem}[1]{\hyperref[lem:#1]{Lemma~\ref*{lem:#1}}}
\newcommand{\defn}[1]{\hyperref[def:#1]{Definition~\ref*{def:#1}}}
\newcommand{\prop}[1]{\hyperref[prop:#1]{Proposition~\ref*{prop:#1}}}
\newcommand{\cor}[1]{\hyperref[cor:#1]{Corollary~\ref*{cor:#1}}}
\newcommand{\fig}[1]{\hyperref[fig:#1]{Figure~\ref*{fig:#1}}}
\newcommand{\tab}[1]{\hyperref[tab:#1]{Table~\ref*{tab:#1}}}
\newcommand{\alg}[1]{\hyperref[alg:#1]{Algorithm~\ref*{alg:#1}}}
\newcommand{\app}[1]{\hyperref[app:#1]{Appendix~\ref*{app:#1}}}
\newcommand{\conj}[1]{\hyperref[conj:#1]{Conjecture~\ref*{conj:#1}}}
\newcommand{\chap}[1]{\hyperref[chap:#1]{Chapter~\ref*{chap:#1}}}

\DeclareMathAlphabet{\mathbbold}{U}{bbold}{m}{n}

\DeclareMathOperator{\DS}{DS}



%% file: introduction.tex

\section{Introduction}

The \emph{$n$-fold direct product} of a function $f \colon \{0,1\}^m \to \{0,1\}$ is the function $f^n \colon \{0,1\}^{m \times n} \to \{0,1\}^n$ defined by
\if\conf1
    \[
	   f^n(x^{(1)},\ldots,x^{(n)}) = \big( f(x^{(1)}), \ldots, f(x^{(n)}) \big)
    \]
\else
    \[
	   f^n(x^{(1)},x^{(2)},\ldots,x^{(n)}) = \big( f(x^{(1)}), f(x^{(2)}), \ldots, f(x^{(n)}) \big)
    \]
\fi
for all $x^{(1)},\ldots,x^{(n)} \in \{0,1\}^m$.
How does the complexity of computing $f^n$ relate to the complexity of $f$?
This question has received a lot of attention over the years in multiple different computational models, including in the setting of randomized query complexity~\cite{IRW94,NRS99,Sha03,KSdW07,JKS10,Dru12,BB19}; let us briefly make this question more precise and present the state of the art on it.

Let $\R_\gamma(f)$ denote the \emph{worst-case randomized query complexity} of $f$ for the success probability parameter $\gamma \in (0,1]$.
This complexity measure represents the minimum number of bits of the input that a randomized algorithm must query (in the worst case over both the input and the internal randomness of the algorithm) in order to compute the value of the function on this input correctly with probability at least $\gamma$.
The randomized query complexity of $f^n$ is bounded above by 
\begin{equation} \label{eq:DPT-wcUB}
	\R_{\gamma^n}(f^n) \le n \, \R_\gamma(f)
\end{equation}
since we can always compute $f^n$ with success probability at least $\gamma^n$ by running $n$ independent instances of an algorithm that computes $f$ with success probability at least $\gamma$.
How tight is this na\"ive upper bound?

In a groundbreaking result, Drucker~\cite{Dru12} showed that every function $f$ satisfies
\begin{equation}
\label{eq:drucker}
\R_{\gamma^n}(f^n) \ge \Omega\left( \min\{(\gamma - \tfrac12)^3, 1-\gamma\} \cdot n \, \R_\gamma(f) \right).
\end{equation}
This result shows that the bound~\cref{eq:DPT-wcUB} is asymptotically tight in the \emph{bounded-error regime} where $\gamma$ is bounded away from $\frac12$ and $1$.
But it leaves open the possibility that the same bound can be improved in the \emph{small-error regime} where $\gamma = 1 - o(1)$ and the \emph{small-advantage regime} where $\gamma = \frac12 + o(1)$.
And as it turns out, the bound~\cref{eq:DPT-wcUB} can indeed be strengthened in those regimes.

Let $\barR_\gamma(f)$ denote the \emph{average-case query complexity} of $f$, where the cost of a randomized algorithm on some input $x$ is measured in expectation over the internal randomness of the algorithm.
With standard truncation and success amplification arguments, we can show that when $n$ is large enough,
\begin{equation}
\label{eq:DPT-acUB}
\R_{\gamma^n}(f^n) \le O \left( n \, \barR_\gamma(f) \right).
\end{equation}
Clearly, we always have $\barR_\gamma(f) \le \R_\gamma(f)$ so this bound is asymptotically always at least as strong as~\cref{eq:DPT-wcUB}.
In the bounded-error regime, $\barR_\gamma(f) = \Theta(\R_\gamma(f))$ so the two bounds are asymptotically equivalent.
But in the small-error regime, there are partial and total functions $f$ for which $\barR_\gamma(f) = o(\R_\gamma(f))$ so~\cref{eq:DPT-acUB} is asymptotically stronger in this regime.
And, as a strong direct sum theorem of Blais and Brody~\cite{BB19} proved, that bound is asymptotically optimal in the special case where $\gamma > 1 - 1/n$:
in this setting, every function $f$ satisfies
\begin{equation}
\label{eq:bb}
\R_{\gamma^n}(f^n) \ge \Omega\left( n \, \barR_{\gamma}(f) \right).
\end{equation}
Since this result only applies to the small-error regime, 
it raises the question of whether there can be a single argument that yields a direct product theorem that is asymptotically optimal in both the bounded-error and small-error regimes.
And, more importantly, it leaves open the possibility that~\cref{eq:DPT-acUB} can also be strengthened in the small-advantage regime.

\subsection{An Optimal Direct Product Theorem}

Our first main result is an optimal direct product theorem which shows that~\cref{eq:DPT-acUB} can indeed be further strengthened and gives a tight characterization of $\R_{\gamma^n}(f^n)$ in terms of the complexity of $f$.
Let us first introduce the complexity measure that gives this characterization.

For a distribution $\mu$ on the domain of $f$, define $\avgR^\mu_\gamma(f)$ to be the minimum expected query cost of a randomized algorithm that succeeds with probability at least $\gamma$ when \emph{both} the expected cost and success probability are measured over the internal randomness of the algorithm \emph{and} the choice of the input $x \sim \mu$.
The \emph{maximum distributional randomized query complexity} measure of $f$ is defined to be $\avgR_\gamma(f) = \max_\mu \avgR^\mu_\gamma(f)$.

The maximum distributional complexity measure $\avgR(\cdot)$ was studied in previous work of Vereshchagin~\cite{Ver98}, where he implicitly shows that this measure is asymptotically equivalent to $\barR(\cdot)$ in the small-error regime and to $\R(\cdot)$ in the bounded-error regime.
As we show in \cref{sec:max-dist-relations}, however, the maximum distributional complexity measure can be asymptotically smaller than these other two measures in the small-bias regime.
Also, and most importantly, this is the complexity measure that characterizes $\R_{\gamma^n}(f^n)$:

\begin{theorem}[Optimal Direct Product Theorem]
\label{thm:DPT}
For every (possibly partial) function $f \colon \{0,1\}^m \to \{0,1\}$, success probability $\gamma$ in the range $\frac12 < \gamma < 1$, and parameter $n \ge 1$,
\begin{equation} \label{eqn:DPT-lb}
\R_{\gamma^n}(f^n) = \Omega\left( n \cdot \avgR_\gamma(f) \right).
\end{equation}
This lower bound is best possible when $n$ is large enough: there is a universal constant $c$ such that for every function $f$, success probability $\gamma$ in the range $\frac12 < \gamma < 1$, and parameter $n \ge \frac{c}{(\min\{2\gamma-1, 1-\gamma\})^{2}}$,
\begin{equation} \label{eqn:DPT-ub}
\R_{\gamma^n}(f^n) = O\left( n \cdot \avgR_\gamma(f) \right).
\end{equation}
\end{theorem}

The bounds in~\cref{thm:DPT} have three different implications for the different regimes:

\begin{description}
	\item[Bounded-error regime:] 
    \if\conf1 \phantom{(Hack fix)} \\ \fi
    In this regime, $\avgR_\gamma(f) = \Theta( \R_\gamma(f) )$ so that \cref{thm:DPT} simplifies to $\R_{\gamma^n}(f^n) = \Theta( n \, \R_\gamma(f) )$, recovering Drucker's direct product theorem~\cite{Dru12}.

	\item[Small-error regime:] 
    \if\conf1 \phantom{(Hack fix)} \\ \fi
    When $\gamma = 1-\epsilon$ for some $\epsilon = o(1)$, $\avgR_{1-\epsilon}(f) = \Theta( \barR_{1-\epsilon}(f) )$ so \cref{thm:DPT} simplifies to $\R_{(1-\epsilon)^n}(f^n) = \Theta( n \, \barR_{1-\epsilon}(f) )$, recovering (and generalizing) the strong direct sum theorem of Blais and Brody~\cite{BB19}.

\item[Small-advantage regime:] 
\if\conf1 \phantom{(Hack fix)} \\ \fi
When $\gamma = \frac{1+\delta}2$ for some $\delta = o(1)$, there are functions $f$ for which $\avgR_\gamma(f) = o( \barR_\gamma(f) )$.
Therefore, the characterization in terms of maximum distributional complexity is necessary for an optimal direct product theorem; no such characterization in terms of average-case query complexity can hold.
\end{description}

We will discuss the maximum distributional complexity measure in more detail in~\cref{sec:max-dist-complexity}.

\subsection{List-Decoding Direct Product Theorem}

\Cref{thm:DPT} characterizes the complexity of computing the function $f^n$ exactly in terms of the complexity of computing $f$.
We can also instead consider many different other hardness amplification questions by requiring algorithms to only ``partially'' compute $f^n$.
Let us mention two such examples:

\begin{description}
\item[Threshold Direct Product:]
\if\conf1 \phantom{(Hack fix)} \\ \fi
Given a function $f^n$ and a parameter $k \le n$, 
what is the complexity of algorithms that compute the value $f(x^{(i)})$ correctly on at least $k$ of the $n$ inputs?
More precisely, define the \emph{$k$-of-$n$ threshold relation} $\Thr^n_k \colon \{0,1\}^n \to \mathcal{P}(\{0,1\}^n)$ to be 
\if\conf1
\begin{align*}
    	\Thr^n_k(y) = \big\{ z &\in \{0,1\}^n : \\
        &|\{i \in [n] : y_i = z_i\}| \ge k \big\}.
\end{align*}
\else
\[
	\Thr^n_k(y) = \left\{ z \in \{0,1\}^n : |\{i \in [n] : y_i = z_i\}| \ge k \right\}.
\]
\fi
What is the worst-case randomized query complexity of $\Thr^n_k \circ f$?

\item[Labelled Threshold Direct Product:]
\if\conf1 \phantom{(Hack fix)} \\ \fi
Given a function $f^n$ and a parameter $k \le n$,
what is the complexity of algorithms that compute the value of $f$ correctly on $k$ of the $n$ inputs \emph{when the algorithms must label the $k$ instances it has solved}?
This task is not equivalent to the threhold direct product question because in this task, no errors are allowed.
We can formalize it by introducing the \emph{$k$-of-$n$ labelled threshold} relation $\Label^n_k \colon \{0,1\}^n \to \mathcal{P}(\{0,1,*\}^n)$ where
\if\conf1
\begin{align*}
    \Label^n_k(y) = \big\{ &z \in \{0,1,*\}^n : \\
    &|\{i \in [n]: z_i \neq *\}| = k \ \wedge \\
    &\forall i \in [n] (z_i \neq *) \Rightarrow z_i = y_i \big\}.
\end{align*}
\else
\[
	\Label^n_k(y) = \left\{ z \in \{0,1,*\}^n : |\{i \in [n]: z_i \neq *\}| = k \wedge \forall i \in [n] (z_i \neq *) \Rightarrow z_i = y_i \right\}.
\]
\fi
What is the worst-case randomized query complexity of $\Label^n_k \circ f$?
\end{description}

The threshold direct product question has been studied in query complexity and other computational models~\cite{IK10,Ung09,Dru12}.
Meanwhile, the labelled threshold direct product has received much less attention so far, in part because analyzing the benefit that an algorithm receives from the ability to choose which instance to solve appears to require very different techniques than the ones used in the threshold direct product question.
In fact, even the special case of the labelled threshold direct product question where $k=1$ is quite interesting; it corresponds to the Choice Problem~\cite{BBKW14,MT18} and it is closely related to the Correlated Samples problem introduced by Bassilakis et al.~\cite{BDG+20}.

These two variants of the threshold direct product relation and other natural ``partial computation'' questions are captured by a general direct product question related to the list-decoding task first considered in the context of error-correcting codes.
For any parameters $n$ and $\ell \le n$, define the \emph{list-decoding} relation $\LD^n_\ell \colon \{0,1\}^n \to \mathcal{P}( \binom{\{0,1\}^n}{2^{n-\ell}} )$ to be 
\[
	\LD_\ell^n(y) = \{ S \subseteq \{0,1\}^n : |S| = 2^{n-\ell} \wedge y \in S\}.
\]
An algorithm correctly computes $\LD_\ell^n \circ f$ on some input $x \in \{0,1\}^{m \times n}$ when it outputs a list of $2^{n-\ell}$ vectors in $\{0,1\}^n$ that includes the vector $(f(x^{(1)}),\ldots,f(x^{(n)}))$.

The randomized query complexity of $\LD^n_\ell \circ f$ is bounded above by
\[
\R_{\gamma^\ell}( \LD^n_{\ell} \circ f ) \le \ell \, \R_\gamma(f)
\]
since we can always compute the value of the first $\ell$ instances of $f$ with success probability $\gamma$ each.
We show that this upper bound is asymptotically tight in the bounded-error regime. 
Writing $\R(f) = \R_{\frac23}(f)$ to denote the randomized query complexity of $f$ in the bounded-error regime, we obtain the following lower bound.

\begin{restatable}[List-Decoding Direct Product Theorem]{theorem}{LDDPT}
\label{thm:LD-DPT}
There exists a universal constant $c \in (\frac12, 1)$ such that
for every (possibly partial) function $f \colon \{0,1\}^m \to \{0,1\}$, $n \ge 1$, $\ell \in \{1,2,\ldots,n\}$, and $\gamma \in [c, 1)$,
\[
\R_{\gamma^\ell}\left(\LD_{\ell}^n \circ f\right) = \Omega\left(  \ell \, \R(f) \right).
\]
\end{restatable}

\Cref{thm:LD-DPT} immediately gives a strong threshold direct product theorem.

\begin{corollary}[Threshold Direct Product Theorem]
\label{cor:ThrDPT}
There exists a universal constant $c \in (\frac12, 1)$ such that for every (possibly partial) function $f \colon \{0,1\}^m \to \{0,1\}$, $n \in \mathbb{N}$, $k \in \{\frac n2 + 1,\ldots,n\}$, and $\gamma \in [c,1)$, if we let 
\if\conf1
\begin{align*}
    \tau_{n,k} &= n - \log|\{x \in \{0,1\}^n : \|x\| \le n-k\}| \\
    &= \log \left( \frac{2^n}{\sum_{i=0}^{n-k} \binom{n}{i}} \right),
\end{align*}
\else
\[
\tau_{n,k} = n - \log|\{x \in \{0,1\}^n : \|x\| \le n-k\}| = \log \left( \frac{2^n}{\sum_{i=0}^{n-k} \binom{n}{i}} \right),
\]
\fi
then 
\[
	\R_{\gamma^{\tau_{n,k}}}\left(\mathrm{Thr}^n_{k} \circ f\right) = \Omega \left( \tau_{n,k}\, \R(f) \right).
\]
\end{corollary}

In particular, \cref{cor:ThrDPT} implies that for every constant $\delta > 0$, there is a constant $c_2 = c_2(\delta)$ such that for all $\gamma > c_2$, the complexity of computing more than a $\frac12 + \delta$ fraction of the $n$ instances correctly with even exponentially small success probability is bounded below by
\[
	\R_{\gamma^n}\left(\mathrm{Thr}^n_{(\frac12 + \delta)n} \circ f\right) = \Omega \big( n \,\R(f) \big).
\]
This result also follows from Drucker's threshold direct product theorem~\cite{Dru12}.
But \Cref{cor:ThrDPT} also applies to settings where we require the algorithm to compute only ``slightly'' more than $\frac12$ of the instances correctly. 
Namely, it implies that for every $1 \le t \le \sqrt{n}$, there is constant $c_3 = c_3(t,n)$ such that for all $\gamma > c_3$,
\[
	\R_{\gamma^{t^2}}\left(\mathrm{Thr}^n_{\frac n2 + t \sqrt{n}} \circ f\right) = \Omega \big( t^2 \,\R(f) \big).
\]
For example, with $t(n) = \sqrt{\log(n)}$, we obtain the conclusion that even computing $\frac n2 + \sqrt{n \log n}$ of the $n$ instances of $f$ correctly with even just a polynomially-small success probability requires query complexity $\Omega( \log n \cdot \R(f))$.

\Cref{thm:LD-DPT} also immediately gives a direct product theorem for the Labelled Threshold problem.

\begin{corollary}[Labelled Threshold Direct Product Theorem]
\label{cor:LabelDPT}
There exists a universal constant $c \in (\frac12, 1)$ such that for every (possibly partial) function $f \colon \{0,1\}^m \to \{0,1\}$, $n \in \mathbb{N}$, $k \le n$, and $\gamma \in [c,1)$, 
\[
	\R_{\gamma^k}\left(\mathrm{Label}^n_{k} \circ f\right) = \Omega \left( k \, \R(f) \right).
\]
\end{corollary}

This result shows that computing any constant fraction $\alpha > 0$ of the $n$ instances of $f$ with exponentially small success probability $\gamma^{\alpha n}$ still has cost $\Omega( n \, \R(f) )$. In the other extreme setting, with $k=1$, this result shows that the ``one-of-$n$'' problem where the algorithm is free to pick which of the $n$ instances it wants to solve still requires query complexity $\Omega( \R(f) )$.

\subsection{Discounted Score}

Our direct product theorems \cref{thm:DPT,thm:LD-DPT} are obtained by taking a new approach.
Namely, instead of trying to prove such theorems directly for randomized query complexity, we start by considering a different notion of complexity of functions.

The \emph{discounted score} of the randomized algorithm $R$ with respect to the function $f$ and the distribution $\mu$ on the domain of $f$ is
\[
	\ds_{f,\alpha}^\mu(R) = \E_{\substack{D \sim R \\ x \sim \mu}} \left[ \ssf\big( D(x) \big) \cdot e^{- \alpha \cdot \cost( D(x) )}\right],
\]
where the parameter $\alpha$ is some (tuneable) \emph{discount factor}, $D(x)$ represents the leaf of $D$ reached by input $x$, and for now we can think of  $\ssf\big(D(x)\big)$ as being the indicator function of whether $D$ outputs the correct value $f(x)$ on input $x$.
The \emph{maximum discounted score} of $f$ with respect to $\mu$ is
\[
	\maxDS^\mu_\alpha(f) = \max_R \ds_{f,\alpha}^\mu(R).
\]

Note that the discounted score is fundamentally different from standard measures of complexity in that it integrates the cost and the score at a ``local'' per-input level and maximizes this value over all randomized algorithms.
This is in contrast to the usual measures of complexity where one of the cost or score is fixed as a constraint, and then the complexity measure seeks to optimize the other parameter among all the algorithms that satisfy the constraint.

\paragraph{Tangential remarks.}
The name \emph{discounted score} is chosen to highlight the conceptual similarity between this measure and the notion of discounted rate in economics.
In this setting, the discounted rate quantifies the idea that \$100 now does not have the same value as \$100 received some time in the future (e.g., because \$100 received now could be invested to accrue some interest during the next year).
As such, it is natural that the discount on the value should be exponential in the delay in receiving this value.
Similarly, if we think of the cost of an algorithm as the time we wait before an answer, the discounted score represents the notion that receiving a correct answer to the value of $f(x)$ now is better (or more valuable) than receiving this answer after waiting for a long time.

We do not use the analogy between discounted score and discounted rate in our arguments, but it is interesting to note that previous work on randomized
query complexity have used a ratio between cost and score~\cite{BB23,BBGM22}. 
This approach is equivalent to discounting by
$1/\cost$ instead of $\exp(-\cost)$, which is analogous
to the ``hyperbolic discounting fallacy'' in behavioral economics.
(See, e.g., \cite{Ain12} and the references therein.)
In other words, our results can be interpreted informally as support for the observation that if one uses the economist-endorsed exponential discounting instead of the behaviorally-common hyperbolic discounting, one can get better direct product theorems for randomized query complexity.

The notion of discounted score also appears to be closely related to other well-studied notions in different fields as well. In particular, is also closely related to the \emph{attenuation} of signals in physics. 
Here the notion of cost can be associated with the distance of the leaf to the root of a tree and the discounted score represents the strength of a signal from the leaf at the root of the tree when all edges have the same length. Once again, we do not use the analogy directly in our arguments but perhaps these analogies suggest that further study of discounted score measures may find other useful applications in computer science beyond the ones presented here.

\paragraph{Motivation.}
Our study of the notion of discounted score is not motivated by either of the analogies discussed above but rather by the fact that it arises naturally when we consider direct product questions for randomized algorithms.
Consider the na\"ive upper bound~\cref{eq:DPT-wcUB} for computing $f^n$ by running $n$ independent instances of a randomized algorithm for $f$.
The success probability of the resulting algorithm behaves multiplicatively---it is the $n$th power of the original success probability---whereas the cost of running both algorithms behaves additively---it is only $n$ times the cost of the original algorithm.

The discrepancy between the behavior of the score and cost parameters appears to be problematic in the study of direct product questions.
Intuitively, it seems like the analysis would be much easier in some cases if both parameters behaved identically (i.e., both were multiplicative or both were additive).
But we can make both parameters behave in the same way if we consider the exponent of the cost instead of the cost itself.
(We would also get this identical behavior if we instead considered the logarithm of the score instead of the score itself, though this approach doesn't seem quite as natural when the score is an indicator function for correctness).

\paragraph{Tensorization.}
The intuition that motivated the study of discounted score for direct product questions is justified: the measure of discounted score it is particularly well suited to the study of product functions.
In fact, it satisfies a ``perfect'' direct product theorem property or, to say this slightly differently, it \emph{tensorizes}.

\begin{restatable}[Tensorization Lemma]{lemma}{tensorization}
\label{lem:DPT-maxDS}
For any (possibly partial) function $f \colon \{0,1\}^m \to \{0,1\}$, any distribution $\mu$ on $\{0,1\}^m$, and any parameter $n \ge 1$,
\[
	\maxDS_\alpha^{\mu^n}(f^n) = \maxDS_\alpha^\mu(f)^n.
\]
\end{restatable}

We prove \cref{lem:DPT-maxDS} establishing inequalities in both directions.
The lower bound $\maxDS_\alpha^{\mu^n}(f^n) \ge \maxDS_\alpha^\mu(f)^n$ 
is easily obtained by analyzing the algorithm for $f^n$ constructed by running independent instances of an algorithm $A$ for $f$ with optimal discounted score on each of the $n$ inputs.

The upper bound $\maxDS_\alpha^{\mu^n}(f^n) \le \maxDS_\alpha^\mu(f)^n$ is more intricate. 
It uses an embedding argument, where we use an algorithm $B$ for $f^n$ that has large discounted score to solve $f$ by embedding the input to $f$ as the $i$th instance to $f^n$ for a randomly chosen $i \in [n]$ and simulating $B$ on this embedded input and $n-1$ other ``fake'' inputs.
The embedding argument in the proof of \cref{lem:DPT-maxDS} departs from previous embedding arguments in that the fake inputs are not generated statically (according to $\mu$ or some other cleverly chosen distribution).
Instead, we simulate $B$ on a query-by-query basis.
Whenever $B$ queries a bit of the real input, we make the corresponding query to provide its answer.
And when $B$ queries one of the fake inputs, we generate the answer to this query according to some distribution $\nu$ conditioned on the bits that we have observed so far.
By choosing $\nu$ to be the $\mu^n$ distribution conditioned on $B$ being correct on all inputs, we obtain exactly the desired lower bound.

\subsection{Success-conditioned score}

In order to obtain a direct product theorem from \cref{lem:DPT-maxDS}, we need to relate the discounted score measure to other measures of randomized query complexity.
We do so by going through another variant of randomized query complexity,
which we call \emph{success-conditioned} randomized
query complexity and denote by $\sR_\gamma(f)$. 

Roughly speaking, $\sR_\gamma(f)$ is defined
similarly to the maximum distributional complexity $\avgR(f)$, with the one important difference that the cost of the algorithm is measured \emph{conditioned on the algorithm succeeding}.
We introduce the notion formally in \cref{sec:score-weighted-defns}.
But for now let us instead introduce some of its important properties.
First, it satisfies
$\sR_\gamma(f)=\Theta(\avgR_\gamma(f))$ over all
$\gamma\in(1/2,1)$ whenever $f$ is a Boolean-valued function.
For functions with non-Boolean output alphabets
(such as $f^n$), however, the success-conditioned query complexity is not equivalent to $\avgR(f)$ and has better behavior. 

Intuitively, conditioning on success helps eliminate the following
type of degenerate randomized algorithm: an algorithm
may attempt to solve $f^n$ with small success probability
$\gamma^n$ by querying everything with probability $\gamma^n$,
and querying nothing (guessing randomly instead) with probability
$1-\gamma^n$. Since the input size grows linearly with $n$
and the desired success decreases exponentially with $n$,
the expected cost of this algorithm is less than $1$ for large
$n$ (even though the worst-case cost is very high).
This means that expected cost measures such as $\barR_{\gamma^n}(f^n)$
and $\avgR_{\gamma^n}(f^n)$ are trivial
for all $f$ (but worst-case cost measures like $\R_{\gamma^n}(f^n)$ are not).
Our introduction of $\sR$ as a success-conditioned
cost allows us to give an expected-cost
version of randomized query complexity for which
the direct product problem is nontrivial.

The direct product theorem is obtained from \lem{DPT-maxDS}
by relating 
success-conditioned complexity to both maximum distributional complexity and maximum discounted score in the following ways.

\begin{restatable}[Equivalence Lemma]{lemma}{maxDSsR}
\label{lem:maxDS-sR}
For every (possibly partial) function $f \colon \{0,1\}^m \to \{0,1\}$, every distribution $\mu$ over the domain of $f$, every success parameter $\gamma$, and every discount factor $\alpha \ge 0$,
\[
	\sR_\gamma^\mu(f) \ge \frac1{\alpha} \log \frac{\maxDS_{\alpha}^\mu(f)}{\gamma}. 
\]
Moreover, this lower bound is best possible: for every $f$, $\mu$, and $\gamma$, there exists a discount factor $\alpha^*$ for which
\[
	\sR_\gamma^\mu(f) = O\left( \frac1{\alpha^*} \log \frac{\maxDS_{\alpha^*}^\mu(f)}{\gamma} \right).
\]
\end{restatable}

The tensorization and equivalence lemmas immediately give an optimal direct product theorem for maximum score-weighted distributional complexity.

\begin{restatable}[Direct Product Theorem for sR]{theorem}{DPTsR}
\label{thm:DPT-sR}
For every (possibly partial) function $f \colon \{0,1\}^m \to \{0,1\}$, every success parameter $\gamma$ in the range $\frac12 \le \gamma \le 1$, and every $n \ge 1$,
\[
\sR_{\gamma^n}(f^n) = \Omega \left( n \cdot \sR_\gamma(f) \right).
\]
\end{restatable}

Note that \cref{thm:DPT-sR} is the only result of its type that we know of where we have the same randomized query complexity measure on both sides of the inequality, perfect product behavior on the success probability parameter, and no loss on the right-hand side that is a function of this parameter. 
As we have already discussed above, it is known that no such theorem can hold for the measures $\R(\cdot)$, $\barR(\cdot)$, or $\avgR(\cdot)$.

\medskip
The lower bound of \Cref{thm:DPT} is obtained as a direct corollary of \cref{thm:DPT-sR} when we observe that worst-case randomized complexity always upper bounds the success-conditioned maximum distributional complexity of any function and that for Boolean-valued functions $f$, the complexity measures $\avgR_\gamma(f)$ and $\sR_\gamma(f)$ can differ by at most a multiplicative factor of 2.

\medskip
To summarize, the measures of randomized query complexity that we have discussed
satisfy
\[\R_\gamma(f)\ge\barR_\gamma(f)\ge
\sR_\gamma(f)\ge \avgR_\gamma(f).\]
For Boolean-valued functions, the right-most two collapse.
If, additionally, $\gamma$ is bounded away from $1/2$,
the rightmost three collapse.
Finally, if $\gamma$
is also bounded away from $1$, then all of these measures coincide. 
A direct product
theorem can have $\sR_\gamma$ on both sides (that is, measuring
both $f^n$ and $f$ with the same notion of query complexity, $\sR$). An even better measure,
which has a perfect direct product theorem, is 
$\DS_\alpha^\mu(f)$,
but it is parametrized by a tricky-to-interpret discount
factor $\alpha$ instead of by the success $\gamma$.

\subsection{Overview of the Other Proofs}

\subsubsection{Direct Product Theorem Upper Bound}

For the upper bound in \cref{thm:DPT}, our goal is to give an algorithm
for computing $n$ copies of $f$ with success probability
$\gamma^n$ using a total worst-case number of queries
of only $O(n\,\avgR_\gamma(f))$. The basic idea is to run
an algorithm for $f$ on each of the $n$ copies, but several
issues arise when we try to implement this idea.

First, the measure $\avgR_\gamma(f)$ does not have a minimax
theorem. 
This means there is no algorithm which achieves such
low query complexity on all distributions.
Instead,
a single algorithm for $f$ can only guarantee that it achieves this expected cost and expected success probability $\gamma$ on a single distribution $\mu$.
To mitigate this limitation, we can use
a minimax theorem for $\R_{\gamma^n}(f^n)$ instead, since it is
simply worst-case query complexity. This gives us a hard
distribution $\mu$ for $f^n$
(which may not be a product distribution).
We then let $\mu_1$ be the marginal distribution of $\mu$
restricted to the first copy of $f$, and run an $\avgR_\gamma(f)$
type algorithm for $f$ with respect to $\mu_1$ on the first copy.
Next, we let $\mu_2$ be the marginal for the second copy,
but conditioned on the bits of the first instance that have been revealed by the algorithm that we ran on
the first copy of $f$. Similarly, $\mu_3$ will be the marginal
of $\mu$ on the third copy of $f$, conditioned on the input bits revealed by the queries of the previous two algorithms.
Since each of these distribution is well-defined, we can find
algorithms $R_1,R_2,\dots,R_n$ for $f$ with respect to these
distributions, such that all of them have success at least $\gamma$
and expected cost at most $\avgR_\gamma(f)$.

Having run these algorithms, our success probability for all copies is $\gamma^n$
(against $\mu$),
as desired. Our query bound is $n\avgR_\gamma(f)$, also as desired.
However, there is another issue: this upper bound on the number
of queries only holds in expectation against $\mu$, rather
than in the worst case. To address this issue, we can impose a cutoff:
terminate the algorithm if the number of queries surpasses
its expectation by a factor of $10$. Such a cutoff will reduce
the success probability, so we want the probability that it happens to
be very low, since we already only have success probability that is exponentially small in $n$. It turns out that when $n$ is large enough, the
runs of the algorithms on different copies of $f$ are sufficiently
independent that the total expected cost has thin tails
(and a standard deviation of $O(\sqrt{n})$). From this we conclude
that the success probability only decreases by the required exponentially small amount.

To finish the argument, we need one more ingredient: success amplification.
We cannot
amplify a generic algorithm for $f^n$: it places probability
less than $1/2$ on the right answer, so even taking a majority vote
might amplify a wrong answer instead. 
So we must amplify
$\avgR_{\gamma}(f)$ and run the whole argument above starting with a
larger value for $\gamma$. 
This approach leads to another complication:
since $\avgR$ is a distributional measure with no minimax theorem,
it is not obvious that its success can be amplified.
In fact, for \emph{fixed} distributional complexity measures, success amplification does not hold.
But, as we show, we do get efficient success amplification for the maximum distributional complexity measure by adapting Schapire's boosting algorithm~\cite{Sch90} to this setting.

\subsubsection{List-Decoding Direct Product Theorem}

For this result, we want to prove a lower bound on the query complexity of an algorithm for $f^n$
which outputs $2^{n-\ell}$ different $n$-bit strings such
that the probability that the correct output string for $f^n$
is in that list is at least $\gamma^\ell$.
Fix a distribution $\mu$ over inputs to $f$, and consider
the behavior of such an algorithm when run on inputs from $\mu^n$.
Since the distribution is fixed,
we can assume that such an algorithm is deterministic,
that is to say, a decision tree.

When the overall success probability of the algorithm is $\gamma^\ell$,
there must be at least a $\gamma^\ell/2$ fraction
of the leaves of the decision tree for which the success probability
(conditioned on reaching this leaf) is at least $\gamma^\ell/2$.
For such a successful leaf $L$, we consider the conditional distribution
of $\mu^n$ conditioned on reaching this leaf. A lemma of Drucker
\cite{Dru12} says that this conditional distribution is still a product
distribution: indeed, it is just the product, over $i\in[n]$,
of the conditional distribution $\mu|_{L_i}$ conditioned on the bits $L_i$ queried
to the $i$-th input on the path to this leaf.

Since $L$ is a successful leaf, we know that there is a list of
$2^{n-\ell}$ strings such that a sample from this product distribution
$\mu|_{L_1} \times \mu|_{L_2} \times \cdots \times \mu|_{L_n}$ is in the list with probability at least
$\gamma^\ell/2$. Without loss of generality, the list in question
can simply be the $2^{n-\ell}$ strings with highest probability
in this product distribution. So a successful leaf places a significant
probability mass on its top $2^{n-\ell}$ most probable strings
(from the posterior distribution $\mu^n|_L$ over $n$-bit strings).
Since the distribution in question is product, this condition can be
translated into an entropy condition: the only way for such high mass
to be placed on so few strings in a product distribution is for the
product distribution to have fairly low entropy: entropy at most
$n-\Omega(\ell)$.

Our situation now is that a decent fraction of the leaves of the decision tree
(at least $\gamma^\ell/2$ of them) are such that conditioned on that leaf,
the posterior distribution over the output string has entropy at most
$n-\Omega(\ell)$.
The next key insight is that we can use the scoring rule $2^{-H(\cdot)}$,
where $H$ is the binary entropy function, to measure the correctness of an algorithm (or a leaf) instead of using success probability. 
The score we give to a leaf $L$ of a decision tree
computing $f^n$ is the product of the scores of the $n$ partial assignments $L_i$
(each with respect to the distribution $\mu$). That is, the score of $L$
is $\prod_{i=1}^n 2^{-H(f(\mu|_{L_i}))}=2^{-H(f^n(\mu^n|_L))}$,
and the score of the decision tree is the expected score of the leaves.

It turns out that \cref{thm:DPT} holds for all scoring functions
which have the multiplicative property described above in which the score of a leaf in a decision
tree for $f^n$ is the product of the scores of the partial assignments $L_i$,
together with a few other simple conditions such as non-negativity.
We can therefore employ our direct product theorem with this new
scoring measure $2^{-H(\cdot)}$, which plays the role of ``success probability''
for our purposes.

The entropy upper bound of $n-\Omega(\ell)$ on a successful leaf $L$
translates into a score lower bound of $2^{\Omega(\ell)-n}$ at that
leaf (this works because the entropy $H(\cdot)$ is additive over the independent
copies of $f$, which translates to $2^{-H(\cdot)}$ being multiplicative
over those copies). Since there are $\Omega(\gamma^\ell)$ fraction of
leaves which have this score lower bound, the overall expected score
of the decision tree is at least $\gamma^\ell 2^{\Omega(\ell)-n}$.
The direct product theorem with this scoring rule now gives a lower
bound on the number of queries used by this decision tree: it is at least
$\Omega(n)$ times the cost of computing a single copy of $f$ with a score
of at least $(\gamma^\ell 2^{\Omega(\ell)-n})^{1/n}$.
This latter score works out to
$2^{-1+\Omega(\ell/n)}$ when $\gamma<1$ a sufficiently large constant.

Since the score function is $2^{-H(\cdot)}$, computing $f$ with a score of
$2^{-1+\Omega(\ell/n)}$ is like attaining binary entropy of at most
$1-\Omega(\ell/n)$ for the posterior distribution over the Boolean output of $f$.
The number of queries required to do this can be related to the
number of queries to compute $f$ with other scoring rules,
such as the ones from \cite{BB23}. One of the scoring rules in that
paper can be linearly amplified. Using this amplification property,
we can take an algorithm which computes $f$ with binary entropy
$1-\Omega(\ell/n)$, repeat it $O(n/\ell)$ times, and obtain
an algorithm which computes $f$ to bounded error. We conclude
that the query cost of obtaining binary entropy $1-\Omega(\ell/n)$
is $\Omega(\R(f)\cdot \ell/n)$, at least when $\mu$ is chosen to
be an appropriate hard distribution.

Putting this together with the direct product theorem,
we obtain a lower bound of $\Omega(n)$ times $\Omega(\R(f)\cdot \ell/n)$,
which gives the desired bound of $\Omega(\ell \R(f))$.


%% file: preliminaries.tex

\section{Definitions and Foundations}

Throughout the paper, the symbol $f$ is reserved for representing a total or partial Boolean function $f \colon X \to \{0,1\}$ for some $X \subseteq \{0,1\}^m$.
The symbol $\mu$ is reserved for distributions over the domain $X$ of $f$.

Since we consider only query complexity,
we associate deterministic algorithms with deterministic decision trees and use both terms interchangeably. 
Every internal node of a decision tree is labelled with the index of a bit of the input, and the outgoing edges from these nodes are labelled with $0$ and $1$. 
These nodes correspond to the queries made by the algorithm.
Every leaf of a decision tree is typically associated with a label in $\{0,1\}$ indicating the output value of the decision tree at the leaf.
(As we will see in the next section, however, in distributional settings we will ignore the value of the label of the leaves when measuring the correctness score of decision trees.)

For each leaf $\ell$ of a decision tree $D$, we use $\ell$ to represent both the leaf itself and the subset of inputs in $X$ that induce a path to $\ell$ in $D$.
For any input $x \in X$, we write $D(x)$ to represent the leaf in $D$ reached by $x$.

A \emph{randomized algorithm} $R$ corresponds to a distribution over deterministic algorithms. 
We write $R$ to denote both the algorithm itself and the associated distribution over deterministic decision trees.
A \emph{leaf} of $R$ is a leaf in one of the deterministic decision trees in its support.

All the definitions carry over to the more general setting of direct product functions $f^n \colon X^n \to \{0,1\}$, $X \subseteq \{0,1\}^m$ as-is, with the single exception that the labels of leaves of decision trees for these functions are labelled with values in $\{0,1\}^n$ or $[0,1]^n$.

\subsection{Score Measures}

The score of a randomized algorithm measures how well it computes a given function $f$.
The most common score considered in the study of randomized algorithms is success probability. 
Our results, however, also apply to other score measures and we will use some of these other score measures for the proof of \cref{thm:LD-DPT} so we introduce a bit of notation to allow the statement of the more general results.

Scoring measures for Boolean functions are defined in terms of a \emph{score function} $\phi \colon [0,1] \to \mathbb{R}$ that satisfies a few properties:
\begin{description}
\item[Boundedness:] 
\if\conf1 \phantom{Hack} \\ \fi
The range of $\phi$ is $[\frac12, 1]$.
\item[Normalization:] 
\if\conf1 \phantom{Hack} \\ \fi
The function satisfies $\phi(0) = \phi(1) = 1$ and $\phi(\frac12) = \frac12$.
\item[Symmetry:] 
\if\conf1 \phantom{Hack} \\ \fi
For every $p \in [0,1]$, $\phi(p)= \phi(1-p)$.
\item[Monotonicity:] 
\if\conf1 \phantom{Hack} \\ \fi
The function $\phi$ is monotonically increasing on $[\frac12,1]$. 
\item[Continuity:] 
\if\conf1 \phantom{Hack} \\ \fi
The function $\phi$ is continuous on $[0,1]$.
\end{description}

The \emph{score} of a leaf $\ell$ of a decision tree $D$ with respect to a target function $f$ and a distribution $\mu$ is defined in terms of the underlying scoring function $\phi$.
Recall that each leaf $\ell$ of $D$ also represents a subset of the domain $X$ of $f$ corresponding to the set of inputs that lead to this leaf. 
We also write $\ell_0 = \ell \cap f^{-1}(0)$ and $\ell_1 = \ell \cap f^{-1}(1)$.
We define the score of $\ell$ to be
\[
\sfmu(\ell) = \phi\left( \frac{\mu(\ell_1)}{\mu(\ell)} \right).
\]

Letting $\ell \sim D(\mu)$ denote the distribution on the leaves of a decision tree $D$ induced by the distribution $\mu$ on the domain of $f$, the score of a randomized algorithm $R$ with respect to $f$ and $\mu$ is 
\[
\barsfmu(R) = \E_{D \sim R} \left[ \E_{\ell \sim D(\mu)}[ \sfmu(\ell) ] \right].
\]
Let us mention a few remarks about these score measures:
\begin{enumerate}
\item For completeness, the notation $\sfmu$ should also include reference to the underlying scoring function $\phi$. We omit it for notational clarity and instead use different notation introduced below whenever we consider a specific score measures.

\item The score measures do not use the values of labels of decision trees when measuring their score, they can instead be thought of measures of discrepancy of the leaves with respect to $f$ and $\mu$. 

\item The score measures can also be obtained from \emph{scoring rules} for forecasting algorithms as introduced in~\cite{BB20,BB23}. These scoring rules can provide additional useful structural properties to the score measures, but all our results only require the basic properties of score functions listed above so the extension to forecasting algorithms is not needed.
\end{enumerate}

We also extend score measures for product functions.
Consider a decision tree $D$ that computes $f^n$ for some function $f \colon \{0,1\}^m \to \{0,1\}$.
Each leaf $\ell$ of $D$ corresponds to a subset of $\{0,1\}^{m \times n}$.
We can always write such a leaf as a product $\ell = \ell^{(1)} \times \cdots \times \ell^{(n)}$ where each $\ell^{(i)} \subseteq \{0,1\}^m$.
We define the score of the leaf $\ell$ of $D$ with respect to a product measure $\mu^n$ on the domain of $f^n$ to be
\[
\sss_{f^n}^{\mu^n}(\ell) = \prod_{i=1}^n \sfmu(\ell^{(i)}).
\]
As before, the score of a randomized algorithm $R$ with respect to $f^n$ and $\mu^n$ is the expected score of its leaves:
\[
\bars_{f^n}^{\mu^n}(R) = \E_{D \sim R} \left[ \E_{\ell \sim D(\mu^n)}[ \sss_{f^n}^{\mu^n}(\ell) ] \right].
\]

\subsubsection{Examples of Score Measures}

We consider the following three score measures to establish \cref{thm:DPT} and \cref{thm:LD-DPT}:

\paragraph{Success probability score.}
The most common way to measure the quality of a randomized algorithm is via its \emph{success probability}. 
The corresponding score measure is obtained by considering the scoring function
\[
\phi_{\success}(p) = \max\{p,1-p\}.
\]

The score of a randomized algorithm $R$ with respect to $\mu$ and $f$ when we use $\phi_{\success}$ is
\[
\barsuccess_f^\mu(R) = \E_{D \sim R}\left[ \E_{\ell \sim D(\mu)} \max\left\{ \frac{\mu(\ell_0)}{\mu(\ell)}, \frac{\mu(\ell_1)}{\mu(\ell)}\right\} \right].
\]
This is the average success probability of $R$ over $\mu$ when each leaf of the trees in $R$'s support are labelled optimally.
In other words, there is a labelling of the leaves of the trees in the support of $R$ for which
\[
\barsuccess_f^\mu(R) = \Pr_{\substack{x \sim \mu \\ D \sim R}}\big[ D \mbox{ outputs } f(x) \mbox{ on input } x\big].
\]

By definition, the score of a randomized algorithm $R$ that computes $f^n$ with respect to the product distribution $\mu^n$ is
\[
\barsuccess_{f^n}^{\mu^n}(R) = \E_{\substack{D \sim R \\ x \sim \mu^n}}\left[ \prod_{i=1}^n \success_f^\mu( D(x)^{(i)} )\right].
\]
By Drucker's product lemma (Lemma 3.2 in \cite{Dru12}), this is exactly the probability that $R$ computes $f^n(x)$ correctly when $x \sim \mu^n$:
\[
\barsuccess_{f^n}^{\mu^n}(R) = \Pr_{\substack{x \sim \mu^n \\ D \sim R}}\big[ D \mbox{ outputs } f^n(x) \mbox{ on input } x\big].
\]

\paragraph{Hellinger score.}

We obtain a different score measure by considering the \emph{Hellinger (squared)} scoring function
\[
\phi_{\hel}(p) = 1 - \sqrt{p(1-p)}.
\]
The Hellinger score of a randomized algorithm $R$ with respect to the function $f$ and distribution $\mu$ is
\if\conf1
\begin{align*}
\barHel_f^\mu(R) 
&= \E_{D \sim R} \left[ \E_{\ell \sim D(\mu)} 1 - \sqrt{\frac{\mu(\ell_0)}{\mu(\ell)} \cdot \frac{\mu(\ell_1)}{\mu(\ell)}}\right] \\
&= 1 - \E_{D \sim R} \sum_{\ell \in D} \sqrt{\mu(\ell_0) \mu(\ell_1)}.
\end{align*}
\else
\[
\barHel_f^\mu(R) 
= \E_{D \sim R} \left[ \E_{\ell \sim D(\mu)} 1 - \sqrt{\frac{\mu(\ell_0)}{\mu(\ell)} \cdot \frac{\mu(\ell_1)}{\mu(\ell)}}\right]
= 1 - \E_{D \sim R} \sum_{\ell \in D} \sqrt{\mu(\ell_0) \mu(\ell_1)}.
\]
\fi

\paragraph{Exponential entropy score.}

A third way to measure the correctness of algorithms is to consider the \emph{exponential entropy} scoring function
\[
\phi_{\ent}(p) = 1 - \tfrac12 2^{-h(p)} = 1 - \tfrac12 p^p(1-p)^{1-p}
\]
where $h \colon p \mapsto -p \log(p) - (1-p) \log(1-p)$ is the binary entropy function.
The exponential entropy score of a randomized algorithm $R$ with respect to $f$ and $\mu$ is
\[
\barEnt_f^\mu(R) = 1 - \frac12 \E_{D \sim R} \left[ \E_{\ell \sim D(\mu)} \mu(\ell) \cdot 2^{-h\left(\frac{\mu(\ell_0)}{\mu(\ell)}\right)} \right].
\]

\subsection{Complexity Measures}

The \emph{cost} of the leaf $\ell$ of the decision tree $D$, denoted $\cost(\ell)$, is the depth of that leaf.
This depth corresponds to the number of queries made by the algorithm before it terminates.
Recall that $D(x)$ denotes the leaf of $D$ reached by input $x$.
The \emph{(absolute) worst-case cost} of the randomized algorithm $R$ is 
\[
	\cost(R) = \max_{\substack{x \in \dom(f) \\ D \in \supp(R)}} \cost( D(x) ).
\]
The \emph{worst-case average cost} of $R$ is
\[
	\barcost(R) = \max_{x \in \dom(f)} \E_{D \sim R}[ \cost(D(x)) ].
\]

The score and cost definitions introduced above give rise to a variety of different complexity measures on functions.
In all of the following definitions, fix any scoring measure $\ssf$.
The \emph{worst-case} randomized query complexity of $f$ is
\[
\R_{\gamma}(f) = \min_{R \,:\, \ssf(R) \ge \gamma} \cost(R)
\]
and the \emph{average-case} complexity of $f$ is
\[
\barR_\gamma(f) = \min_{R \,:\, \ssf(R) \ge \gamma} \barcost(R).
\]
For the reader's convenience, we review some of the basic facts relating worst-case and average-case query complexity in \cref{sec:wc-ac}.
Our results also rely on additional notions of randomized query complexity. 
We introduce these alternative measures in the next two subsections.

\subsection{Maximum Distributional Complexity}
\label{sec:max-dist-complexity}

The \emph{average cost} of $R$ over the distribution $\mu$ is 
\[
	\barcost^\mu(R) = \E_{x\sim \mu, D \sim R}[ \cost(D(x))].
\]
The \emph{distributional query complexity} of $f$ with respect to the distribution $\mu$ and the score parameter $\gamma$ is
\[
\avgR^\mu_\gamma(f) = \min_{R \,:\, \barsfmu(R) \ge \gamma} \barcost^\mu(R),
\]
the minimum expected cost of a randomized algorithm with expected score $\gamma$, where both the cost and score are taken in expectation over the internal randomness of the algorithm and the choice of the input $x \sim \mu$.

The \emph{maximum distributional randomized query complexity} of $f$ for the score parameter $\gamma$ is
\[
\avgR_\gamma(f) = \max_\mu \avgR^\mu_\gamma(f) = \max_\mu \min_{R \,:\, \barsfmu(R) \ge \gamma} \barcost^\mu(R).
\]
This measure of complexity is distribution-independent.

Note that because both the score and the cost of the randomized algorithms are taken in expectation over both the internal randomness and the choice of input over the underlying distribution, standard minimax theorems do not apply and so (unlike in the more common distributional settings considered in the literature where one of these measures is taken in a worst-case fashion) it is not clear whether $\avgR_\gamma(f)$ should always equal its deterministic analogue $\avgD_\gamma(f)$ or not.

\subsubsection{Score Amplification Properties}

The standard success amplification and boosting arguments used when considering $\R(f)$ and $\barR(f)$ do not apply for the maximum distributional complexity measure.
Nevertheless, we show that the standard arguments can be extended to show that this complexity measure does satisfy the usual score amplification properties.

First, we show that Schapire's boosting algorithm~\cite{Sch90} can be extended to give a success amplification theorem analogous to the ones for worst-case and average-case complexity for maximum distributional complexity.

\begin{restatable}{lemma}{boosting}
\label{lem:boosting-maxdist}
For any $\delta > 0$, any $f \colon \{0,1\}^m \to \{0,1\}$, and any success probability parameter $\gamma \in [\frac{1+\delta}2,1 - \delta]$,
\[
\avgR_{\gamma + \frac{\delta}3}(f) \le 4 \cdot \avgR_{\gamma}(f).
\]
\end{restatable}

We also prove a (stronger) linear score amplification lemma for the Hellinger score measure.

\begin{restatable}{lemma}{linamplification}
\label{lem:hellinger-lin-amp}
For any $\delta \in [0,1]$ and any $f \colon \{0,1\}^m \to \{0,1\}$,
\[
	\R(f) = O \left( \frac{1}\delta \cdot \avgR_{\Hel\,\frac{1+\delta}2}(f) \right).
\]
\end{restatable}

Note that these two score amplification properties do \emph{not} hold in the standard distributional setting where we consider distributional complexity measures over a fixed distribution (c.f.,~\cite{Sha03,BKST24}). 
The proofs of \cref{lem:boosting-maxdist,lem:hellinger-lin-amp} strongly rely on the fact that the maximum distributional complexity measure holds over all distributions.
These proofs are presented in \cref{sec:amplification}.

\subsection{Score-Weighted Complexity}
\label{sec:score-weighted-defns}

In addition to the complexity measures introduced above, we also introduce one more notion of complexity which is based on a different way to measure the cost of algorithms.

\subsubsection{Sucess-Conditioned and Score-Weighted Cost}

The average cost of a randomized algorithm $R$ with respect to a distribution $\mu$ is measured on average over all inputs $x$.
We can also measure the expected cost of $R$ only on the inputs for which it outputs the correct value, giving a \emph{success-conditioned} average cost
\[
\E_{\substack{x \sim \mu \\ D \sim R}}\big[ \cost( D(x) ) \mid D \mbox{ outputs } f(x) \big].
\]
This success-conditioned cost of randomized algorithms can also be generalized to a \emph{score-weighted cost} measure for any scoring rule by setting
\[
	\scost_f^\mu(R) 
	= \frac{\E_{x \sim \mu, D \sim R}\big[ \ssf(D(x)) \cdot \cost( D(x) ) \big]}{\barsfmu(R)}.
\]
As we can verify directly, this measure coincides with the success-conditioned average cost when we set $\ssf = \success_f$ and use the optimal labels on the decision tree leaves.

The average cost and score-weighted cost of a randomized algorithm $R$ can differ substantially.
Consider for example a function $f \colon \{0,1\}^m \to \{0,1\}$ and an algorithm $R$ for $f^n$ that queries all $mn$ bits of the input to determine the value of $f^n(x)$ with probability $p = \frac{1}{2^n-1}$ and otherwise guesses the value of $f^n(x)$. 
For any distribution $\mu$, the average cost of this algorithm is $\barcost^\mu(R) = pmn = \frac{mn}{2^n-1}$ but its score-weighted cost with the success score function is
\if\conf1
\begin{align*}
\scost_{f^n}^\mu(R) &= \frac{p \cdot mn  + (1-p)2^{-n}\cdot 0}{p + (1-p)2^{-n}} \\
&= \frac{mn}{2} \gg \barcost^\mu(R).
\end{align*}
\else
\[
\scost_{f^n}^\mu(R) = \frac{p \cdot mn  + (1-p)2^{-n}\cdot 0}{p + (1-p)2^{-n}} = \frac{mn}{2} \gg \barcost^\mu(R).
\]
\fi

\subsubsection{Maximum Score-weighted Distributional Query Complexity}

The notion of score-weighted cost gives rise to one more distributional measure of complexity, the \emph{distributional score-weighted query complexity}:
\[
\sR_\gamma^\mu(f) = \min_{R \,:\, \barsfmu(R) \ge \gamma} \scost_f^\mu(R).
\]
Once again, we obtain a distribution-free complexity measure by taking the maximum distributional score-weighted query complexity over all choices of $\mu$:
\[
\sR_\gamma(f) = \max_{\mu} \sR_\gamma^\mu(f). 
\]

The maximum score-weighted distributional query complexity is related to both the maximum distributional and worst-case randomized query complexity measures in the following ways.

\begin{restatable}{proposition}{sRandAvgR}
\label{prop:sR-avgR}
For every function $f \colon \{0,1\}^m \to \{0,1\}$ and score parameter $\gamma$,
\[
	\tfrac12 \cdot \avgR_\gamma(f) \le \sR_\gamma(f) \le \frac{1}{\gamma} \cdot \avgR_\gamma(f)
\]
\end{restatable}

\begin{restatable}{proposition}{sRandR}
\label{prop:sR-R}
For every function $f \colon \{0,1\}^m \to \{0,1\}$, parameter $n \in \mathbb{N}$, and score $\gamma$,
\[
	\sR_{\gamma^n}(f) \le \R_{\gamma^n}(f).
\]
\end{restatable}

We complete the proofs of these two propositions in \cref{sec:relations-sR}.

%% file: directproduct.tex

\section{Proof of the Direct Product Theorem}

We complete the proofs of the standard Direct Product Theorem, \cref{thm:DPT}, as well as the closely related \cref{thm:DPT-sR} in this section.
We begin by establishing the lower bounds for both of these theorems.

\subsection{Proof of the Lower Bounds}

As discussed in the introduction, there are two main technical results at the heart of the proofs of the Direct Product Theorems.
The first is the tensorization lemma for discounted score, which we restate here for the reader's convenience.

\tensorization*

The second is the equivalence lemma which shows the relation between discounted score and score-weighted distributional complexity.
We also restate it for convenience.

\maxDSsR*

We complete the proof of \cref{lem:DPT-maxDS} in \cref{sec:tensorization} and the proof of \cref{lem:maxDS-sR} in \cref{sec:equivalence}.
But first, let us show how these lemmas let us complete the proof of the lower bound for the Direct Product Theorem for maximum score-weighted distributional complexity.

\DPTsR*

\begin{proof}
Choose a distribution $\mu$ on the domain of $f$ that maximizes $\sR_\gamma^\mu(f)$, so that $\sR_\gamma(f) = \sR_\gamma^\mu(f)$.
Then if we choose $\alpha^*$ to be the parameter for which the second inequality in~\cref{lem:DPT-maxDS} holds,
\begin{align*}
\sR_{\gamma^n}^{\mu^n}(f^n)
&\ge \frac{1}{\alpha^*} \log \frac{\maxDS_{\alpha^*}^{\mu^n}(f^n)}{\gamma^n} & \mbox{(\if\conf0 First bound in \fi\cref{lem:maxDS-sR})} \\
&= \frac{1}{\alpha^*} \log \frac{\maxDS_{\alpha^*}^{\mu}(f)^n}{\gamma^n} & \mbox{(\cref{lem:DPT-maxDS})} \\
&= n \cdot \frac{1}{\alpha^*} \log \frac{\maxDS_{\alpha^*}^\mu(f)}{\gamma} \\
&= \Omega\left( n \cdot \sR_\gamma^\mu(f) \right) & \mbox{(\if\conf0 Second bound in \fi\cref{lem:maxDS-sR})}. 
\end{align*}
\end{proof}

And from there we complete the proof of the lower bound for our main Direct Product Theorem.
While this lower bound is stated in terms of success probability in the introduction for simplicity, we show here that it also holds for arbitrary score measures.

\begin{theorem}[Lower bound of \cref{thm:DPT}, generalized]
For every function $f \colon \{0,1\}^m \to \{0,1\}$, any score measure, any score $\gamma \in (\frac12, 1)$, and every $n \ge 1$,
\[
\R_{\gamma^n}(f^n) = \Omega\left( n \cdot \avgR_\gamma(f) \right).
\]
\end{theorem}

\begin{proof}
Combining \cref{thm:DPT-sR} with our earlier results relating the different complexity measures, we obtain 
\begin{align*}
\R_{\gamma^n}(f^n)
&\ge \sR_{\gamma^n}(f^n) & \mbox{(\cref{prop:sR-R})} \\
& \ge \Omega \left( n \, \sR_\gamma(f) \right) & \mbox{(\cref{thm:DPT-sR})} \\
&\ge \Omega \left( n \, \big(\tfrac12 \avgR_\gamma(f)\big) \right) & \mbox{(\cref{prop:sR-avgR})} \\
&= \Omega \left( n \, \avgR_\gamma(f) \right).
\end{align*}
\end{proof}

\subsection{Proof of the Upper Bound} 

As a first step in the proof of the upper bound of \cref{thm:DPT}, we establish a basic result showing that it is possible to bound both the 
average and worst-case cost of algorithms effectively.

\begin{proposition}
\label{prop:bounded-costs}
Fix any $\delta > 0$. For any $\gamma$ in the range $\frac12 + \delta \le \gamma \le 1 - \delta$, $f \colon \{0,1\}^n \to \{0,1\}$, and distribution $\mu$, 
there exists an algorithm $A$ that simultaneously satisfies
\begin{align*}
\barsuccess_f^\mu(A) &\ge \gamma + \tfrac\delta6, \\
\barcost^\mu(A) &\le 4 \cdot \avgR_\gamma(f), \qquad \mbox{and} \\
\cost(A) &\le \tfrac{24}\delta \cdot \avgR_\gamma(f).
\end{align*}
\end{proposition}

\begin{proof}
Let $A_0$ be an algorithm which satisfies $\barsuccess_f^\mu(A_0) \ge \gamma$ and 
$\barcost^\mu(A_0) \le \avgR^\mu_\gamma(f) \le \avgR_\gamma(f)$.
By the success amplification \cref{lem:boosting-maxdist}, there exists an algorithm $A_1$ with $\barsuccess_f^\mu(A_1) \ge \gamma + \frac\delta3$ and $\barcost^\mu(A_1) \le 4 \cdot \barcost^\mu(A_0) \le 4 \cdot \avgR_\gamma(f)$.

Define $A$ to be the algorithm that simulates $A_1$ but aborts and guesses the value of $f$ on the input whenever $A_1$'s cost exceeds $\frac{6}{\delta} \cdot \barcost(A_1)$.
By construction, the average- and worst-case costs of $A$ satisfy the conclusion of the proposition.
And by Markov's inequality, the probability that $A$ aborts $A_1$'s execution is at most $\frac\delta6$, so its success probability is at least $(\gamma + \frac\delta3) - \frac\delta6 = \gamma + \frac\delta6$.
\end{proof}

We also use the following variant of the Azuma--Hoeffding inequality.

\begin{lemma}
\label{lem:azuma}
Let $(X_1,\ldots,X_n)$ be a sequence of random variables over $\mathbb{R}^n$ that satisfies
$\E[ X_{k} \mid X_1,\ldots,X_{k-1} ] \le \alpha$ and
$|X_{k}| \le c$
for all $k = 1,\ldots,n$. Then
\[
	\Pr\left[ \sum_{i=1}^n X_i > \alpha n + t \right] \le e^{-t^2/2c^2n}.
\]
\end{lemma}

\begin{proof}
The sequence $(Y_0,Y_1,\ldots,Y_n)$ with $Y_0 = 0$ and $Y_k = \sum_{i=1}^k (X_i - \alpha)$ forms a submartingale with bounded differences $|Y_{k+1} - Y_k| \le c$, and so by the standard form of the Azuma--Hoeffding inequality,
\[
\Pr\left[ \sum_{i=1}^n X_i > \alpha n + t \right]
= \Pr\left[ Y_n > t \right] \le e^{-t^2/2c^2n}. \qedhere
\]
\end{proof}

We are now ready to complete the proof of the upper bound for the Direct Product Theorem.

\begin{theorem}[Upper Bound of \cref{thm:DPT}]
For any function $f \colon \{0,1\}^m \to \{0,1\}$, success parameter $\gamma$ in the range $\frac12 < \gamma < 1$, and $n \ge \frac{2(24)^2}{(\min\{2\gamma-1, 1-\gamma\})^{2}}$,
\begin{equation}
\R_{\gamma^n}(f^n) = O\left( n \cdot \avgR_\gamma(f) \right).
\end{equation}
\end{theorem}

\begin{proof}
Fix $\delta = \min\{2\gamma - 1, 1-\gamma\}$.
By Yao's Minimax Principle, it suffices to show that for any fixed distribution $\mu$, there is an algorithm with worst-case cost $O(n \, \avgR_\gamma(f) )$ that computes $f^n$ with score at least $\gamma^n$ on inputs drawn from the distribution $\mu$.
Fix $\mu$.
Consider the algorithm $B_0$ that computes the value of $f^n$ on inputs drawn from $\mu$ by 
simulating algorithms $A_1,\ldots,A_n$ on the $n$ instances of the input in the following way.

First, let $\mu_1$ denote the marginal of $\mu$ on the first instance of the input.
By \cref{prop:bounded-costs}, there exists an algorithm $A_1$ with
$\barsuccess_f^{\mu_1}(A_1) \ge \gamma + \frac\delta6$, $\barcost^{\mu_1}(A_1) \le 4 \cdot \avgR_\gamma(f)$, and $\cost(A_1) \le \frac{24}{\delta} \avgR_\gamma(f)$.
The algorithm $B_0$ executes $A_1$ on the first instance of the input.
Then for $i=2,3,\ldots,n$ considered in order, let $\mu_i$ denote the marginal of $\mu$ on the $i$th instance of the input conditioned on the bits of the first $i-1$ instances that have been revealed by the algorithms $A_1,\ldots,A_{i-1}$ that have been previously executed.
There exists an algorithm $A_i$ that satisfies 
$\barsuccess_f^{\mu_i}(A_i) \ge \gamma + \frac\delta6$, $\barcost^{\mu_i}(A_i) \le 4 \cdot \avgR_\gamma(f)$, and $\cost(A_i) \le \frac{24}{\delta} \avgR_\gamma(f)$.
The algorithm $B_0$ executes $A_i$ on the $i$th instance of the input.

By the guarantees on $A_1,\ldots,A_n$, the score of $B_0$ is
\[
	\barsuccess_{f^n}^\mu(B_0) = \prod_{i=1}^n \barsuccess_f^{\mu_i}(A_i)
	\ge \left( \gamma + \tfrac\delta6 \right)^n
\]
and its expected cost satisfies
\[
	\barcost^\mu(B_0) = \sum_{i=1}^n \barcost^{\mu_i}(A_i) \le 4n \cdot \avgR_\gamma(f).
\]
To complete the proof of the upper bound in \cref{thm:DPT}, we must also bound the worst-case cost of our algorithm for $f^n$. Unfortunately, the cost of $B_0$ can be as large as $\frac{24}{\delta} n \cdot \avgR_\gamma(f)$, which is larger than our target bound by a factor of $O(\frac1\delta)$.

So let us now define $B$ to be the algorithm that simulates $B_0$ but aborts whenever the total cost of the simulation exceeds $9 n \cdot \avgR_\gamma(f)$.
The algorithm $B$ outputs any value whenever it aborts.
Let us write $E$ to denote the event that $B$ aborts, and let $W$ denote the event that $B_0(x) = f^n(x)$.
The probability that $B$ succeeds on inputs drawn from $\mu$ is bounded below by
\begin{align*}
\Pr_{x \sim \mu}\left[ B(x) = f^n(x)\right]
&\ge \Pr_{x \sim \mu}\left[ W \wedge \overline{E} \right] \\
&= \Pr_{x \sim \mu}\left[ W \right] \cdot \Pr_{x \sim \mu}\left[ \overline{E} \,\mid\, W \right].
\end{align*}
We have already seen that $\Pr_{x \sim \mu}\left[ W \right] = \barsuccess_{f^n}^{\mu^n}(B_0) \ge (\gamma + \frac\delta6)^n$.
It remains to bound the second term.

Define $X_i$ to be a random variable that represents the cost of $A_i$ when we run $B_0$.
The expected value of $X_i$ conditioned on $X_1,\ldots,X_{i-1}$ and $W$ satisfies
\if\conf1
\begin{align*}
\E_{x \sim \mu}\big[ X_i \mid &X_1,\ldots,X_{i-1},W \big] \\
&\le \max_{\mu_i} \E_{x_i \sim \mu_i}\left[ X_i \,\mid\, A_i(x_i) = f(x_i) \right] \\
&\le \max_{\mu_i} \frac{\E_{x_i \sim \mu_i}[ X_i ]}{\Pr_{x_i \sim \mu_i}[ A_i(x_i) = f(x_i)]} \\
&\le \max_{\mu_i} \frac{\barcost^{\mu_i}(A_i)}{\barsuccess_f^{\mu_i}(A_i)} 
\end{align*}
\else\begin{align*}
\E_{x \sim \mu}\left[ X_i \mid X_1,\ldots,X_{i-1},W \right]
&\le \max_{\mu_i} \E_{x_i \sim \mu_i}\left[ X_i \,\mid\, A_i(x_i) = f(x_i) \right] \\
&\le \max_{\mu_i} \frac{\E_{x_i \sim \mu_i}[ X_i ]}{\Pr_{x_i \sim \mu_i}[ A_i(x_i) = f(x_i)]} \\
&\le \max_{\mu_i} \frac{\barcost^{\mu_i}(A_i)}{\barsuccess_f^{\mu_i}(A_i)} 
\end{align*}
\fi

Since the success probability score of any algorithm is bounded below by $\barsuccess_f^{\mu_i}(A_i) \ge \frac12$, we conclude that
\if\conf1
\begin{align*}
\E_{x \sim \mu}\left[ X_i \mid X_1,\ldots,X_{i-1},W \right]
&\le 2\max_{\mu_i} \barcost^{\mu_i}(A_i) \\
&\le 8 \, \avgR_\gamma(f).
\end{align*}
\else
\[
\E_{x \sim \mu}\left[ X_i \mid X_1,\ldots,X_{i-1},W \right]
\le 2\max_{\mu_i} \barcost^{\mu_i}(A_i)
\le 8 \, \avgR_\gamma(f).
\]
\fi
And each $X_i$ is bounded by $0 \le X_i \le \frac{24}{\delta} \avgR_\gamma(f)$.
By \cref{lem:azuma} with parameters $t = n \avgR_\gamma(f)$ and $c = \frac{24}{\delta} \avgR_\gamma(f)$, we then have 
\if\conf1
\begin{align*}
	\Pr[ E \,\mid\, W ] &\le \Pr\big[ X > \E[X \mid W] + n \cdot \avgR_\gamma(f) \mid W \big] \\
	&\le e^{-\frac{1}{2 \cdot (24)^2} \delta^2 n}.
\end{align*}
\else
\[
	\Pr[ E \,\mid\, W ] \le \Pr\big[ X > \E[X \mid W] + n \cdot \avgR_\gamma(f) \mid W \big]
	\le e^{-\frac{1}{2 \cdot (24)^2} \delta^2 n}.
\]
\fi
When $n > 2(24)^2/\delta^2$, the last expression is smaller than $e^{-1}$.
And in this case the overall success of $B$ is at least $(\gamma + \frac{\delta}6)^n e^{-1} \ge \gamma^n \cdot (1 + \frac{\delta}6)^n e^{-1} \ge \gamma^n e^{\frac{\delta}{12}n - 1}$, which is at least $\gamma^n$ when $n \ge 12/\delta$.
\end{proof}

%% file: tensorization.tex

\section{Proof of the Tensorization Lemma}
\label{sec:tensorization}

\tensorization*

\begin{proof}
We start by proving that $\maxDS_\alpha^{\mu^n}(f^n) \le \maxDS_\alpha^\mu(f)^n$.
Let $A$ be a randomized algorithm with discounted score $\ds_{f,\alpha}^\mu(A) = \maxDS_\alpha^\mu(f)$.
Consider the algorithm $B$ that computes $f^n$ by running $n$ independent instances of $A$ sequentially on each of the $n$ inputs.
By the independence of the $n$ inputs and of the $n$ instances of $A$, the discounted score of $B$ over $\mu^n$ is
\if\conf1
\begin{align*}
&\ds_{f^n,\alpha}^{\mu^n}(B) \\
&= \E_{\substack{x \sim \mu^n \\ D \sim A^n}}\left[ \prod_{i=1}^n \sfmu(D^{(i)}(x_i)) \cdot e^{-\alpha \, \cost(D^{(i)}(x_i))} \right] \\
&= \prod_{i=1}^n \E_{\substack{x_i \sim \mu \\ D^{(i)} \sim A}}\left[ \sfmu(D^{(i)}(x_i)) \cdot e^{-\alpha \, \cost(D^{(i)}( x_i))} \right] \\
&= \ds_{f,\alpha}^\mu(A)^n.
\end{align*}
\else
\begin{align*}
\ds_{f^n,\alpha}^{\mu^n}(B) 
&= \E_{\substack{x_1,\ldots,x_n \sim \mu \\ D^{(1)},\ldots,D^{(n)} \sim A}}\left[ \prod_{i=1}^n \sfmu(D^{(i)}(x_i)) \cdot e^{-\alpha \, \cost(D^{(i)}(x_i))} \right] \\
&= \prod_{i=1}^n \E_{\substack{x_i \sim \mu \\ D^{(i)} \sim A}}\left[ \sfmu(D^{(i)}(x_i)) \cdot e^{-\alpha \, \cost(D^{(i)}( x_i))} \right] \\
&= \ds_{f,\alpha}^\mu(A)^n.
\end{align*}
\fi
Therefore, $\maxDS_\alpha^{\mu^n}(f^n) \le \ds_{f^n,\alpha}^{\mu^n}(B) = \ds_{f,\alpha}^\mu(A)^n = \maxDS_\alpha^\mu(f)^n$, as we wanted to show.

\medskip

We now prove that $\maxDS_\alpha^{\mu^n}(f^n) \ge \maxDS_\alpha^\mu(f)^n$.
Let $B$ be an algorithm that computes $f^n$ and has discounted score $\ds_{f^n, \alpha}^{\mu^n}(B) = \maxDS_\alpha^{\mu^n}(f^n)$.
Without loss of generality, we can assume that $B$ is deterministic.
We define a randomized algorithm $A$ for $f$ in the following way.

The algorithm $A$ is defined with respect to a specific distribution $\nu$ on the domain of $f^n$ that depends on $f$, $\mu$, $D$, and $\alpha$ in a way that will be specified later.
The behaviour of $A$ is defined as follows.
First, it picks $i \in [n]$ uniformly at random.
Next, it runs the algorithm $A_i$ which works by running $B$ on $n-1$ fake inputs and the real input to $f$, with the real input placed in position $i$.
When $B$ makes a query inside the $i$th input to $f$, the algorithm $A_i$ queries the real input.
And when $B$ makes a query to the $j$th input for some $j \neq i$, the answer to the query is generated artificially by the algorithm $A_i$; it generates the query answer from the conditional distribution $\nu$ conditioned on the transcript of $B$ so far.
(That is, it samples a string $y$ from $\nu$ conditioned on the string being consistent with all the queries answered so far, and it determines the value of the query by making the corresponding query to $y$.)

Let $\pi_i$ be the distribution over the leaves of $D$ reached by this simulation by $A_i$, assuming the real input (the one at position $i$) was sampled from $\mu$.
Extend $\pi_i$ to a distribution over the domain of $f^n$ by replacing each leaf of $B$ with a sample from $\mu^n$ conditioned on reaching that leaf. 
The behaviour of $A_i$ on an input sampled from $\mu$ is then the same as the behaviour of $B$ on an input sampled from $\pi_i$.
For an algorithm acting on $n$ inputs to $f$, let $\cost_i$ denote the number of queries to the $i$th input.
We now write
\if\conf1
\begin{align*}
\ds_{f,\alpha}^\mu(A) 
	&= \E_{i \in [n]} \left[ \sum_{\ell \in B} \pi_i(\ell) \cdot \sfmu(\ell_i) \, e^{-\alpha \, \cost_i(\ell)} \right] \\
	&= \sum_{\ell \in B} \E_{i \sim [n]}\left[ \pi_i(\ell) \cdot \sfmu(\ell_i) \, e^{-\alpha \, \cost_i(\ell)} \right].
\end{align*}
\else
\[
	\ds_{f,\alpha}^\mu(A) 
	= \E_{i \in [n]} \left[ \sum_{\ell \in B} \pi_i(\ell) \cdot \sfmu(\ell_i) \, e^{-\alpha \, \cost_i(\ell)} \right]
	= \sum_{\ell \in B} \E_{i \sim [n]}\left[ \pi_i(\ell) \cdot \sfmu(\ell_i) \, e^{-\alpha \, \cost_i(\ell)} \right].
\]
\fi
The positivity of the score measures guarantees that each term of the sum is non-negative.
By applying the AM-GM inequality to each term of the sum, the last expression is bounded below by
\if\conf1
\begin{align*}
\sum_{\ell \in B} &\E_{i \sim [n]}\left[ \pi_i(\ell) \cdot \sfmu(\ell_i) \, e^{-\alpha \, \cost_i(\ell)} \right] \\
&\ge \sum_\ell \left( \prod_{i=1}^n \pi_i(\ell) \cdot \sfmu(\ell_i) \, e^{-\alpha \, \cost_i(\ell)} \right)^{1/n}.
\end{align*}
\else
\[
\sum_{\ell \in B} \E_{i \sim [n]}\left[ \pi_i(\ell) \cdot \sfmu(\ell_i) \, e^{-\alpha \, \cost_i(\ell)} \right]
\ge \sum_\ell \left( \prod_{i=1}^n \pi_i(\ell) \cdot \sfmu(\ell_i) \, e^{-\alpha \, \cost_i(\ell)} \right)^{1/n}.
\]
\fi
We next observe that $\sum_i \cost_i(\ell) = \cost(\ell)$.
By definition, we also have $\prod_i \sfmu(\ell_i) = \sss_{f^n}^{\mu^n}(\ell)$.
And then the key observation at the heart of the proof is that we also have
\[
	\prod_{i=1}^n \pi_i(\ell) = \nu(\ell)^{n-1} \mu^n(\ell).
\]
To see why this last identity holds, note that we can write $\pi_i(\ell) = \prod_{(v,w)} \pi_i(w \mid v)$ where the product is over all the directed edges $(v,w)$ from the root of the tree to the leaf $\ell$. By our construction, 
\[
\pi_i(w|v) = \begin{cases}
	\mu^n(w|v) & \mbox{if node $v$ queries the $i$th instance} \\
	\nu(w|v) & \mbox{otherwise}.
\end{cases}
\]
So we have
\if\conf1
\begin{align*}
	\prod_{i=1}^n \pi_i(\ell) 
	&= \prod_{(v,w)} \prod_{i} \pi_i(w|v) \\
	&= \prod_{(v,w)} \mu^n(w|v) \nu(w|v)^{n-1} \\
	&= \prod_{(v,w)} \mu^n(w|v) \left( \prod_{(v,w)} \nu(w|v) \right)^{n-1} \\
	&= \mu^n(\ell) \nu(\ell)^{n-1}.
\end{align*}
\else
\begin{align*}
	\prod_{i=1}^n \pi_i(\ell) 
	= \prod_{(v,w)} \prod_{i} \pi_i(w|v) 
	&= \prod_{(v,w)} \mu^n(w|v) \nu(w|v)^{n-1} \\
	&= \prod_{(v,w)} \mu^n(w|v) \left( \prod_{(v,w)} \nu(w|v) \right)^{n-1}
	= \mu^n(\ell) \nu(\ell)^{n-1}.
\end{align*}
\fi

Combining all of the above, we then end up with
\if\conf1
\begin{align*}
	\ds_{f,\alpha}^\mu&(A) \ge \\
	&\sum_\ell \left( \nu(\ell)^{n-1} \mu^n(\ell) \cdot \sss_{f^n}^{\mu^n}(\ell) \, e^{-\alpha \, \cost(\ell)} \right)^{1/n}.
\end{align*}
\else
\[
	\ds_{f,\alpha}^\mu(A) 
	\ge \sum_\ell \left( \nu(\ell)^{n-1} \mu^n(\ell) \cdot \sss_{f^n}^{\mu^n}(\ell) \, e^{-\alpha \, \cost(\ell)} \right)^{1/n}.
\]
\fi
Define $\nu$ by setting
\[
	\nu(y) = \frac{\mu^n(y) \cdot \sss_{f^n}^{\mu^n}(B(y)) \, e^{-\alpha \, \cost(B(y))}}{\sum_z \mu^n(z)\cdot \sss_{f^n}^{\mu^n}(B(z)) \, e^{-\alpha \, \cost(B(z))}}.
\]
This gives the corresponding distribution on leaves defined by
\[
	\nu(\ell) = \frac{\mu^n(\ell) \cdot \sss_{f^n}^{\mu^n}(\ell) \, e^{-\alpha \, \cost(\ell)}}{\sum_{\ell'} \mu^n(\ell') \cdot \sss_{f^n}^{\mu^n}(\ell') \, e^{-\alpha \, \cost(\ell')}},
\]
and so we obtain
\if\conf1
\begin{align*}
	\ds_{f,\alpha}^\mu(A) 
	&\ge \left( \sum_{\ell} \mu^n(\ell) \cdot \sss_{f^n}^{\mu^n}(\ell) \, e^{-\alpha \, \cost(\ell)} \right)^{1/n} \\
	&= \ds_{f^n,\alpha}^{\mu^n}(B)^{1/n}. \qedhere
\end{align*}
\else
\[
	\ds_{f,\alpha}^\mu(A) 
	\ge \left( \sum_{\ell} \mu^n(\ell) \cdot \sss_{f^n}^{\mu^n}(\ell) \, e^{-\alpha \, \cost(\ell)} \right)^{1/n} 
	= \ds_{f^n,\alpha}^{\mu^n}(B)^{1/n}. \qedhere
\]
\fi
\end{proof}

%% file: equivalence.tex

\section{Proof of the Equivalence Lemma}
\label{sec:equivalence}

\subsection{Proof Overview}

Both the upper and lower bounds in \cref{lem:maxDS-sR} rely on an alternative characterization of discounted score and score-weighted cost in terms of a reweighted distribution on leaves.
Let us first introduce this distribution.

For any randomized algorithm $R$, let us denote the distribution $\pi_R$ on the leaves of the deterministic decision trees in the support of $R$ obtained by taking the probability that we reach this leaf when $D \sim R$ and $x \sim \mu$ and rescaling it by the weight $\barsfmu(\ell)$. 
Namely, letting $\mu(x)$ and $R(D)$ denote the probability that $x \sim \mu$ and $D \sim R$, respectively, we define the rescaled probability of the leaf $\ell$ of $D$ to be
\[
	\pi_R(\ell) 
= \frac{\mu(\ell) R(D) \cdot \barsfmu(\ell)}{\barsfmu(R)}.
\]

The discounted score of the randomized algorithm $R$ satisfies
\[
	\ds_{f,\alpha}^\mu(R) 
	= \barsfmu(R) \cdot \E_{\ell \sim \pi_R}[ e^{-\alpha \, \cost(\ell)}]
\]
and the success-conditioned cost of an algorithm satisfies
\[
	\scost_f^\mu(R) 
	= \E_{\ell \sim \pi_R}[ \cost(\ell) ].
\]
This means that we can express the maximum discounted score of $f$ as
\[
	\maxDS_\alpha^\mu(f) 
	= \max_{R} \barsfmu(R) \cdot \E_{\ell \sim \pi_R}[ e^{-\alpha \, \cost(\ell)}]
\]
and its distributional success-conditioned query complexity of $f$ can be expressed as
\[
	\sR_\gamma^\mu(f) = \min_{R \,:\, \barsfmu(R) \ge \gamma} \E_{\ell \sim \pi_R}[ \cost(\ell) ].
\]
With this characterization,
the lower bound in \cref{lem:maxDS-sR} is an immediate consequence of Jensen's inequality.

The upper bound requires a few more ingredients.
The first is \cref{prop:success-discount}, showing that for every success parameter $\gamma$, there exists a discount parameter $\alpha^*$ for which the maximum discounted score $\maxDS_{\alpha^*}^\mu(f)$ is achieved by a randomized algorithm $R$ with average score $\barsfmu(R)$ exactly equal to $\gamma$.
Choosing the parameter $\alpha^*$ and corresponding algorithm $R^*$ in the conclusion of the proposition, we obtain the identity
\[
	\maxDS_{\alpha^*}^\mu(f) = \gamma \cdot \E_{\pi_{R^*}}[ e^{-\alpha \, \cost(\ell)} ].
\]
What we would like to use at this point is a ``reverse Jensen inequality'' to obtain a bound of the form
\[
	\E_{\pi_{R^*}}[ e^{-\alpha^* \, \cost(\ell)} ] \stackrel{?}{\le} 
	e^{-c \, \alpha^* \cdot \E_{\pi_{R^*}}[ \cost(\ell)]}
\]
for some constant $c$.
No such inequality holds in general, but we do have the fact that when $y$ is a random variable bounded by some interval $[0,d]$, then convexity of the $e^{-y}$ function implies that $e^{-y} \le (1-\frac{y}d) e^{-0} + \frac yd e^{-d}$ and so
\[
	\E[e^{-y}] \le 1 - \frac{1-e^{-d}}{d} \E[y] \le e^{-\frac{1-e^{-d}}d \E[y]}.
\]
Applying this inequality to the random variable $y = \alpha^* \cost(R^*(x))$ then gives exactly the type of inequality we are aiming for if we have the guarantee that the cost of the leaves in $R^*$ is bounded above by $O(1/\alpha^*)$. 

In general, however, we do not have any such guarantee on the maximum cost of the algorithm $R^*$.
We can get around this limitation by considering a truncation of $R^*$ that yields a different algorithm $R^\dagger$.
The key to completing the argument with this approach is to show that such a truncation does not affect the average cost by more than a constant multiplicative factor.
And this is where the fact that $R^*$ has maximum discounted score for the parameter $\alpha^*$ comes into play: we can argue that any such algorithm cannot have much of its average cost (even in the rescaled distribution $\pi_{R^*}$) on high-cost leaves.
The proof of this last fact is obtained by a contradiction argument, showing that if the last condition holds, we can define an algorithm with strictly better discounted score by truncating part of the tree of $R^*$, contradicting the fact that it has maximal discounted score.

\subsection{The Proof}

\maxDSsR*

\begin{proof}
By the characterization of maximum discounted score introduced in the proof overview,
\begin{align*}
\maxDS_\alpha^\mu(f) 
&= \max_{R} \barsfmu(R) \cdot \E_{\ell \sim \pi_R}[ e^{-\alpha \, \cost(\ell)}] \\
&\ge \max_{R \,:\, \barsfmu(R) \ge \gamma} \barsfmu(R) \cdot \E_{\pi_R}[ e^{-\alpha \, \cost(\ell)}] \\
&\ge \gamma \cdot \max_{R \,:\, \barsfmu(R) \ge \gamma} \E_{\pi_R}[ e^{-\alpha \, \cost(\ell)} ].
\end{align*}
By Jensen's inequality and the convexity of the function $y \mapsto e^{-\alpha y}$, 
for every algorithm $R$ we have the inequality $\E_{\pi_R}[ e^{-\alpha \, \cost(\ell)}] \ge e^{-\alpha \, \E_{\pi_R}[ \cost(\ell) ]}$. 
Therefore,
\begin{align*}
\maxDS_\alpha^\mu(f)
&\ge \gamma \cdot \max_{R \,:\, \barsfmu(R) \ge \gamma} e^{- \alpha \, \E_{\pi_R}[ \cost(\ell) ]} \\
&= \gamma \cdot e^{- \alpha \cdot \min_{R \,:\, \barsfmu(R) \ge \gamma} \E_{\pi_R}[ \cost(\ell) ] } \\
&= \gamma \cdot e^{-\alpha \, \sR_\gamma^\mu(f) }.
\end{align*}
This completes the proof of the lower bound in the lemma.

\bigskip
We now complete the proof of the upper bound.
Let $\alpha^*$ be the discount rate and $R^*$ be the associated algorithm in the conclusion of \cref{prop:success-discount}.
These objects satisfy
\if\conf1
\begin{align*}
	\maxDS_{\alpha^*}^\mu(f) 
	&= \barsfmu(R^*) \cdot \E_{\ell \sim \pi_{R^*}}[ e^{-\alpha^* \, \cost(\ell)} ] \\
	&= \gamma \cdot \E_{\ell \sim \pi_{R^*}}[ e^{-\alpha^* \, \cost(\ell)} ].
\end{align*}
\else
\[
	\maxDS_{\alpha^*}^\mu(f) 
	= \barsfmu(R^*) \cdot \E_{\ell \sim \pi_{R^*}}[ e^{-\alpha^* \, \cost(\ell)} ]
	= \gamma \cdot \E_{\ell \sim \pi_{R^*}}[ e^{-\alpha^* \, \cost(\ell)} ].
\]
\fi

Let $d$ be a depth parameter whose value will be set later.
The trivial inequality $\min\{\cost(\ell),d\} \le \cost(\ell)$ implies that $\maxDS_{\alpha^*}^\mu(f)$ can be bounded above by
\[
	\maxDS_{\alpha^*}^\mu(f) \le \gamma \cdot \E_{\ell \sim \pi_{R^*}}[ e^{-\alpha^* \, \min\{\cost(\ell),d\}} ].
\]
By convexity of the function $e^{-\alpha^* \, y}$, when $y$ is bounded by $0 \le y \le d$, then
\if\conf1
\begin{align*}
	e^{-\alpha^* \, y} 
	&\le \left( 1 - \frac{y}d \right) \cdot e^{-\alpha^* \cdot 0} + \frac yd \cdot e^{-\alpha^* \, d} \\
	&= 1 - \frac{1-e^{-\alpha^* \, d}}d \cdot y.
\end{align*}
\else
\[
	e^{-\alpha^* \, y} 
	\le \left( 1 - \frac{y}d \right) \cdot e^{-\alpha^* \cdot 0} + \frac yd \cdot e^{-\alpha^* \, d}
	= 1 - \frac{1-e^{-\alpha^* \, d}}d \cdot y.
\]
\fi
Applying this inequality to the random variable $y = \min\{\cost(\ell), d\}$, we obtain
\if\conf1
\begin{align*}
	\E_{\ell \sim \pi_{R^*}}&[ e^{-\alpha^* \, \min\{\cost(\ell), d\}} ] \\
	&\le 1 - \frac{1 - e^{-\alpha^* \, d}}d \cdot \E_{\pi_{R^*}}[ \min\{\cost(\ell), d\} ] \\
	&\le e^{-\frac{1-e^{-\alpha^* \, d}}d \cdot \E_{\pi_{R^*}}[ \min\{\cost(\ell), d\}]}.
\end{align*}
\else
\[
	\E_{\ell \sim \pi_{R^*}}[ e^{-\alpha^* \, \min\{\cost(\ell), d\}} ]
	\le 1 - \frac{1 - e^{-\alpha^* \, d}}d \cdot \E_{\pi_{R^*}}[ \min\{\cost(\ell), d\} ]
	\le e^{-\frac{1-e^{-\alpha^* \, d}}d \cdot \E_{\pi_{R^*}}[ \min\{\cost(\ell), d\}]}.
\]
\fi
Setting $d = 2\ln(2)/\alpha^*$, we observe that
\[
	\frac{1-e^{-\alpha^* \, d}}{d} 
	= \frac{1 - 2^{-2}}{2 \ln(2) / \alpha^*} 
	= \frac3{8\ln(2)} \cdot \alpha^*.
\]
Combining all of the above gives us the inequality
\[
	\maxDS_{\alpha^*}^\mu(f) \le \gamma \cdot e^{-c' \alpha^* \E_{\pi_{R^*}}[ \cost(\lfloor D \rfloor_d(x) )]}
\]
for the constant $c' = 3/8\ln(2)$.
To complete the proof of the upper bound, it remains to bound $\E_{\pi_{R^*}}[ \min\{\cost(\ell), d\}]$ from below by a multiplicative constant times $\E_{\pi_{R^*}}[ \cost(\ell) ]$.

\begin{claim}
With $R^*$ and $d$ as defined above, $\E_{\pi_{R^*}}[ \min\{\cost(\ell), d\} ] \ge \frac13 \E_{\pi_{R^*}}[ \cost(\ell) ]$.
\end{claim}

For any depth $\tau$, let us write $X_{> \tau}$ to denote the event that $\cost(\ell) > \tau$ and, in a slight abuse of notation, let us also write $X_{> \tau}$ to denote the indicator random variable corresponding to this event.
Similarly, let $X_{\le \tau}$ denote the event that $\cost(\ell) \le \tau$ and its indicator random variable.

Let us write $L(\mu,R^*)$ to denote the distribution on the leaf $D(x)$ reached when we draw $x \sim \mu$ and $D \sim R^*$.
We can write the discounted score of $R^*$ as
\if\conf1
\begin{align*}
\ds_{f,\alpha^*}^\mu(R^*)
&= \E_{\ell \sim L(\mu,R^*)}[ \barsfmu(\ell) \cdot e^{-\alpha^* \, \cost(\ell)} X_{\le \tau}] \\
& \quad + \E_{\ell \sim L(\mu,R^*)}[ \barsfmu(\ell) \cdot e^{-\alpha^* \, \cost(\ell)} X_{> \tau}] \\
&= \E_{\ell \sim L(\mu,R^*)}[ \barsfmu(\ell) \cdot e^{-\alpha^* \, \cost(\ell)} X_{\le \tau}] \\
	&\quad + \gamma \, \E_{\pi_{R^*}}[e^{-\alpha^* \, \cost(\ell)} X_{> \tau}].
\end{align*}
\else
\begin{align*}
\ds_{f,\alpha^*}^\mu(R^*)
&= \E_{\ell \sim L(\mu,R^*)}[ \barsfmu(\ell) \cdot e^{-\alpha^* \, \cost(\ell)} X_{\le \tau}]
+ \E_{\ell \sim L(\mu,R^*)}[ \barsfmu(\ell) \cdot e^{-\alpha^* \, \cost(\ell)} X_{> \tau}] \\
&= \E_{\ell \sim L(\mu,R^*)}[ \barsfmu(\ell) \cdot e^{-\alpha^* \, \cost(\ell)} X_{\le \tau}]
	+ \gamma \, \E_{\pi_{R^*}}[e^{-\alpha^* \, \cost(\ell)} X_{> \tau}].
\end{align*}
\fi

Let $\lfloor R^* \rfloor_\tau$ be the randomized algorithm by truncating $R^*$ at depth $\tau$.
The truncation is completed by replacing each internal node at depth $\tau$ in a tree $D$ in the support of $R^*$ with a leaf to obtain the tree $\lfloor D \rfloor_\tau$.
For every leaf $\ell$ in a decision tree $D$ in the support of $R^*$, let us write $\lfloor \ell \rfloor_\tau$ to denote the leaf in $\lfloor D \rfloor_\tau$ reached by the inputs that lead to $\ell$ in $D$.
The cost of these leaves satisfies $\cost(\lfloor \ell \rfloor_\tau) = \min\{\cost(\ell),\tau\}$.
And the discounted score of the algorithm $\lfloor R^* \rfloor_\tau$ is bounded below by
\if\conf1
\begin{align*}
\ds_{f,\alpha^*}^\mu(\lfloor R^* \rfloor_\tau)
&\ge \E_{\ell \sim L(\mu,R^*)}[ \barsfmu(\ell) e^{-\alpha^* \, \cost(\ell)} X_{\le \tau}] \\
&\quad + \tfrac12 e^{-\alpha^* \, \tau} \Pr_{\ell \sim L(\mu,R^*)}[X_{> \tau}]
\end{align*}
\else
\[
\ds_{f,\alpha^*}^\mu(\lfloor R^* \rfloor_\tau)
\ge \E_{\ell \sim L(\mu,R^*)}[ \barsfmu(\ell) e^{-\alpha^* \, \cost(\ell)} X_{\le \tau}]
+ \tfrac12 e^{-\alpha^* \, \tau} \Pr_{\ell \sim L(\mu,R^*)}[X_{> \tau}]
\]
\fi
where the second term uses the bounded range property of scoring measures which guarantees that the score of every leaf (including in particular the ones introduced during the truncation process) is at least $\frac12$.
Since $R^*$ has maximum discounted score, we must have $\ds_{f,\alpha^*}^\mu(R^*) \ge \ds_{f,\alpha^*}^\mu(\lfloor R^* \rfloor_\tau)$.
And so the above decompositions of the discounted scores of $R^*$ and $\lfloor R^* \rfloor_\tau$ imply that
\[
	\gamma \, \E_{\ell \sim \pi_{R^*}}[e^{-\alpha^* \, \cost(\ell)} X_{> \tau}]
\ge \tfrac12 e^{-\alpha^* \, \tau} \Pr_{\ell \sim L(\mu,R^*)}[ X_{> \tau} ].
\]
The right-hand side of this inequality can be further bounded below by noting that
\if\conf1
\begin{align*}
	\Pr_{\ell \sim L(\mu,R^*)}[X_{> \tau}] 
	&\ge \E_{\ell \sim L(\mu,R^*)}[ \barsfmu(\ell) X_{> \tau}] \\
    &= \barsfmu(R^*) \E_{\ell \sim \pi_{R^*}}[ X_{> \tau}] \\
	&= \gamma \, \Pr_{\ell \sim \pi_{R^*}}[ X_{> \tau} ].
\end{align*}
\else
\[
	\Pr_{\ell \sim L(\mu,R^*)}[X_{> \tau}] 
	\ge \E_{\ell \sim L(\mu,R^*)}[ \barsfmu(\ell) X_{> \tau}]
    = \barsfmu(R^*) \E_{\ell \sim \pi_{R^*}}[ X_{> \tau}]
	= \gamma \, \Pr_{\ell \sim \pi_{R^*}}[ X_{> \tau} ].
\]
\fi
And the left-hand side can be bounded above by using
\if\conf1
\begin{align*}
	\E_{\ell \sim \pi_{R^*}}&[e^{-\alpha^* \, \cost(\ell)} X_{> \tau}] 
	\le \\
    &\left(\Pr_{\pi_{R^*}}[X_{> \tau}] - \Pr_{\pi_{R^*}}[ X_{> \tau + d}]\right) e^{-\alpha^* \tau} \\
    &+ \Pr_{\pi_{R^*}}[ X_{> \tau + d}] e^{-\alpha^*(\tau + d)}.
\end{align*}
\else
\[
	\E_{\ell \sim \pi_{R^*}}[e^{-\alpha^* \, \cost(\ell)} X_{> \tau}] 
	\le \left(\Pr_{\pi_{R^*}}[X_{> \tau}] - \Pr_{\pi_{R^*}}[ X_{> \tau + d}]\right) e^{-\alpha^* \tau} + \Pr_{\pi_{R^*}}[ X_{> \tau + d}] e^{-\alpha^*(\tau + d)}.
\]
\fi
Combining the last three inequalities, we obtain
\[
	\Pr_{\pi_{R^*}}[X_{> \tau + d}](1-e^{-\alpha^* d}) \le \tfrac12 \Pr_{\pi_{R^*}}[X_{> \tau}]. 
\]
Recalling that $d = 2\ln(2)/\alpha^*$, this gives $\frac34 \Pr_{\pi_{R^*}}[X_{> \tau + d}] \le \frac12 \Pr_{\pi_{R^*}}[X_{> \tau}]$ or, equivalently,
\[
	\Pr_{\pi_{R^*}}[X_{> \tau + d}] \le \frac23 \Pr_{\pi_{R^*}}[ X_{> \tau} ].
\]
We now use this inequality to complete the proof of the claim.

The expected cost of $R^*$ under the distribution $\pi_{R^*}$ can be expressed as
\begin{align*}
	\E_{\pi_{R^*}}[ \cost(\ell) ] 
	= \E_{\pi_{R^*}}[ \cost(\ell) X_{\le d}] 
	+ \E_{\pi_{R^*}}[ \cost(\ell) X_{> d}]. 
\end{align*}
The  latter term can be bounded above by
\if\conf1
\begin{align*}
	\E_{\pi_{R^*}}&[ \cost(\ell) X_{> d}] \\
& \le \sum_{k \ge 1} (k+1)d \left(\Pr_{\pi_{R^*}}[X_{> kd}] - \Pr_{\pi_{R^*}}[ X_{> (k+1)d}\right) \\
&= d \sum_{k \ge 1} \Pr_{\pi_{R^*}}[ X_{> kd} ].
\end{align*}
\else
\begin{align*}
	\E_{\pi_{R^*}}[ \cost(\ell) X_{> d}] 
& \le \sum_{k \ge 1} (k+1)d \left(\Pr_{\pi_{R^*}}[X_{> kd}] - \Pr_{\pi_{R^*}}[ X_{> (k+1)d}\right) \\
&= d \sum_{k \ge 1} \Pr_{\pi_{R^*}}[ X_{> kd} ].
\end{align*}
\fi
From the inequality in the last paragraph, for every $k \ge 2$,
\[
	\Pr_{\pi_{R^*}}[ X_{> kd} ] 
	\le \tfrac23 \Pr_{\pi_{R^*}}[ X_{> (k-1)d} ] 
	\le \left(\tfrac23\right)^{k-1} \Pr_{\pi_{R^*}}[ X_{> d} ].
\]
So $\E_{\pi_{R^*}}[ \cost(\ell) X_{> d}] \le d \sum_{k \ge 0} (\frac23)^k \Pr_{\pi_{R^*}}[X_{>d}]$ and
\[
	\E_{\pi_{R^*}}[ \cost(\ell) ] \le 
	\E_{\pi_{R^*}}[ \cost(\ell) X_{\le d}] + 3d \Pr_{\pi_{R^*}}[X_{> d}].
\]
Meanwhile, the expected cost of $\lfloor R^* \rfloor_d$ under $\pi_{R^*}$ can be expressed as
\[
	\E_{\pi_{R^*}}[ \cost(\lfloor \ell \rfloor_d) ] 
	= \E_{\pi_{R^*}}[ \cost(\ell) X_{\le d}] + d \Pr_{\pi_{R^*}}[ X_{> d}].
\]
Therefore, $\E_{\pi_{R^*}}[ \cost(\ell) ] \le 3 \E_{\pi_{R^*}}[ \cost(\lfloor \ell \rfloor_d) ] = 3 \E_{\pi_{R^*}}[ \min\{\cost(\ell),d\}]$, as we wanted to show.
\end{proof}

%% file: listdecoding.tex

\section{List-Decoding Direct Product Theorems}

\subsection{Preliminary Lemmas}

Before we can complete the proof of \cref{thm:LD-DPT}, we need to establish a few information-theoretic facts.
The \emph{entropy} of the random variable $X \sim \pi$ over a finite domain is
\[
H(X) = \sum_{x} \pi(x) \log \tfrac1{\pi(x)} = \E_{X \sim \pi}[\log \tfrac1{\pi(X)} ].
\]
The \emph{cross-entropy} of $X \sim \pi$ with respect to $Z \sim \sigma$ over the same domain is
\[
H(Z,X) = \E_{X \sim \sigma}[ \log \tfrac1{\pi(X)}].
\]

We start with a basic inequality on the entropy of collections of independent Bernoulli random variables.

\begin{lemma}
\label{lem:HX-low-cross-entropy}
Let $X_1,\ldots,X_n$ and $Z_1,\ldots,Z_n$ be $2n$ independent Bernoulli random variables with $X_i \sim \mathrm{Ber}(\frac{1+\delta_i}2)$ and $Z_i \sim \mathrm{Ber}(\frac{1+\tau_i}2)$ for $i=1,\ldots,n$.
If $X = (X_1,\ldots,X_n)$ and $Z = (Z_1,\ldots,Z_n)$ satisfy
$H(Z) \ge n - \alpha$
and
$H(Z,X) \le n - \beta$
for some $\beta \le \alpha$, then
\[
	H(X) \le n - \frac{\beta^2}{16\alpha}.
\]
\end{lemma}

\begin{proof}
If there is any index $i$ for which $\delta_i > \tau_i$, then the random variable $X'$ drawn as $X$ except with bias $\tau_i$ for $X_i$ also satisfies the two constraints since $H(Z_i,X_i) \ge H(Z_i) = H(Z_i,X'_i)$ and it has strictly larger entropy $H(X') > H(X)$.
Therefore, without loss of generality, we can consider only the case where $\delta_i \le \tau_i$ for each $i=1,\ldots,n$.

The entropy of $X_i$ satisfies
\if\conf1
\begin{align*}
	1-H(X_i) 
    &= 1 - h\left(\frac{1+\delta_i}2\right) \\
    &= \frac1{2 \ln 2}\left( \delta_i^2 + \frac{\delta_i^4}{6} + \frac{\delta_i^6}{15} + \cdots\right).
\end{align*}
\else
\[
	1-H(X_i) = 1 - h\left(\frac{1+\delta_i}2\right) = \frac1{2 \ln 2}\left( \delta_i^2 + \frac{\delta_i^4}{6} + \frac{\delta_i^6}{15} + \cdots\right).
\]
\fi
So $\frac{\delta_i^2}{2 \ln 2} \le 1 - H(X_i) \le \delta_i^2$
and by adding the entropy of each coordinate of $X$,
\[
	\frac{\|\delta\|_2^2}{2\ln 2} \le n - H(X) \le \|\delta\|_2^2.
\]
Similarly, $n - H(Z) \ge \|\tau\|_2^2/2\ln 2$. 
With the condition in the lemma statement, we therefore have the bound
\[
\|\tau\|_2^2 \le 2 \ln 2 \cdot \alpha.
\]

We also have
\[
	1 - H(Z_i,X_i) = 1 - H(X_i) + \frac{\tau_i-\delta_i}2 \log( \frac{1+\delta_i}{1-\delta_i}).
\]
Since $\log(1+x) \le x/\ln(2)$, the last term in this expression is bounded by
\if\conf1
\begin{align*}
	\frac{\tau_i-\delta_i}2 \log( \frac{1+\delta_i}{1-\delta_i} ) &= 
	\frac{\tau_i-\delta_i}2 \log( 1 + \frac{2\delta_i}{1-\delta_i} ) \\
    &\le
	\frac{\tau_i-\delta_i}2 \cdot \frac{2\delta_i}{\ln(2)(1-\delta_i)} \\
    &= \frac{\tau_i-\delta_i}{1-\delta_i} \cdot \frac{\delta_i}{\ln 2} \\ 
    &\le \frac{\tau_i \delta_i}{\ln 2},
\end{align*}
\else
\[
	\frac{\tau_i-\delta_i}2 \log( \frac{1+\delta_i}{1-\delta_i} ) = 
	\frac{\tau_i-\delta_i}2 \log( 1 + \frac{2\delta_i}{1-\delta_i} ) \le
	\frac{\tau_i-\delta_i}2 \cdot \frac{2\delta_i}{\ln(2)(1-\delta_i)} 
    = \frac{\tau_i-\delta_i}{1-\delta_i} \cdot \frac{\delta_i}{\ln 2}\le \frac{\tau_i \delta_i}{\ln 2},
\]
\fi
and we obtain
\if\conf1
\begin{align*}
n - H(Z,X) &\le n - H(X) + \frac{\left< \tau, \delta \right>}{\ln 2} \\
&\le \|\delta\|_2^2 + \frac{\|\tau\|_2 \|\delta\|_2}{\ln 2} \\
&\le \frac{2}{\ln 2}\,\|\tau\|_2 \|\delta\|_2
\end{align*}
\else
\[
	n - H(Z,X) \le n - H(X) + \frac{\left< \tau, \delta \right>}{\ln 2} \le \|\delta\|_2^2 + \frac{\|\tau\|_2 \|\delta\|_2}{\ln 2} \le \frac{2}{\ln 2}\,\|\tau\|_2 \|\delta\|_2
\]
\fi
where the last inequality uses the fact that $\|\delta\|_2 \le \|\tau\|_2 \le \frac{\|\tau\|_2}{\ln 2}$.

The above inequalities and the conditions of the lemma now give
\[
	\beta^2 \le \big(n - H(Z,X)\big)^2 \le \frac{4\|\tau\|_2^2\|\delta\|_2^2}{\ln^2 2} \le \frac{8}{\ln2} \, \alpha \|\delta\|_2^2. 
\]
These inequalities imply that $\frac{\|\delta\|_2^2}{\ln 2} \ge \frac{\beta^2}{8 \, \alpha}$.
Therefore,
\[
	n-H(X) \ge \frac{\|\delta\|_2^2}{2 \ln 2} \ge \frac{\beta^2}{16 \alpha}. \qedhere
\]
\end{proof}

We use the above bound to obtain the following result.

\begin{lemma}
\label{lem:HX-list}
Fix any $n \in \mathbb{N}$, $\ell \in \{1,\ldots,n\}$, and $\epsilon \in [0,\frac14]$.
If $p_1,\ldots,p_n$ are biases and $L \subseteq \{0,1\}^n$ is a list of $2^{n-\ell}$ elements such 
that when $X_i \sim \mathrm{Ber}(p_i)$ are drawn independently, $X = (X_1,\ldots,X_n)$ satisfies
\[
	\Pr[ X \in L] \ge 2^{-\epsilon \ell},
\]
then $H(X) \le n - \frac{(1-\epsilon)^2 \ell}{320}$.
\end{lemma}

\begin{proof}
Fix $\alpha = 2^{-\frac{1-\epsilon}2}$ and $d = \alpha 2^{(1-\epsilon)\ell-n}$.
Consider the disjoint subsets $B_0,B_1,B_2,\ldots$ of $L$ where for each $i \ge 0$,
\[
	B_i = \{ x \in L : 2^i d < p(x) \le 2^{i+1} d \}.
\]
By the hypothesis of the lemma, $\sum_{x \in L} p(x) \ge 2^{-\epsilon \ell}$.
And by construction, the only elements in $L$ that are not in one of the bins $B_i$, $i \ge 0$ are the ones for
which $p(x) \le d = \alpha 2^{(1-\epsilon)\ell -n}$, so the total mass of the excluded elements is at most $2^{n-\ell} \cdot \alpha  2^{(1-\epsilon)\ell - n} = \alpha 2^{-\epsilon \ell}$.
This means that the mass of the points in the union of the $B_i$ sets is bounded below by
\[
\sum_{i \ge 0} \sum_{x \in B_i} p(x) \ge (1-\alpha) 2^{-\epsilon \ell}.
\]
This lower bound implies that there must be an index $i \ge 0$ for which the bin $B_i$ has total mass
\begin{equation}
\label{eq:large-bin}
\sum_{x \in B_i} p(x) \ge (1-\alpha) 2^{-\epsilon \ell} \cdot 2^{-(i+1)}
\end{equation}
otherwise we would have $\sum_{i \ge 0} \sum_{x \in B_i} p(x) < (1-\alpha)2^{-\epsilon \ell} \cdot 2^{-(i+1)} = (1-\alpha) 2^{-\epsilon \ell}$, contradicting the lower bound on the mass of points in the union of the bins.
Fix an index $i$ that satisfies \cref{eq:large-bin} and let $D$ be the uniform distribution on the bin $B_i$.

The upper bound on the mass of each element in $B_i$ and the lower bound on the total mass of this bin imply that the number of points in $B_i$ is bounded below by
\[
|B_i| \ge \frac{(1-\alpha) 2^{-\epsilon \ell} \cdot 2^{-(i+1)}}{2^{i+1}d}
= \frac{1-\alpha}{\alpha} \cdot 2^{n-\ell} \cdot 2^{-2(i+1)}.
\]
So the entropy of $D$ is bounded below by
\[
	H(D) = \log |B_i| \ge n - \big(\ell + 2(i+1) - \log \tfrac{1-\alpha}{\alpha}\big).
\]
When $\epsilon \le \frac14$, then $\log \tfrac{1-\alpha}{\alpha} = \log(2^{-\frac{1-\epsilon}2}-1) \ge -2$, so in this case the lower bound simplifies to $H(D) \ge n - \big( \ell + 4 + 2i \big)$.
Meanwhile, the lower bound on the mass of the elements in $B_i$ implies that
\if\conf1
\begin{align*}
    H(D,X) 
    &= - \mathrm{E}_{x \sim D}[ \log p(x) ] \\
    &\le 
	-\log( 2^i d) \\
    &= -\log( \alpha 2^{(1-\epsilon)\ell - n} 2^i) \\
    &= n - \big((1-\epsilon)\ell - \log \tfrac1\alpha + i \big).
\end{align*}
\else
\[
	H(D,X) = - \mathrm{E}_{x \sim D}[ \log p(x) ] \le 
	-\log( 2^i d) = -\log( \alpha 2^{(1-\epsilon)\ell - n} 2^i) = 
	n - \big((1-\epsilon)\ell - \log \tfrac1\alpha + i \big).
\]
\fi
Since $\log \frac1\alpha = \frac{1-\epsilon}2$, we can simplify this bound to $H(D,X) \le n - \big( (1-\epsilon) \ell - \frac{1-\epsilon}2 + i \big)$.

Consider now the Bernoulli random variables $Z_i \sim \mathrm{Ber}(q_i)$ with parameters $q_i = \Pr_{x \in D}[ x_i = 1]$ for $i=1,\ldots,n$. 
Each $Z_i$ is equal to the marginal $D_i$ of $D$ on the $i$th coordinate.
The random variable $Z$ satisfies
\if\conf1
\begin{align*}
H(Z) = \sum_{i=1}^n H(Z_i) &= \sum_{i=1}^n H(D_i) \\ 
&\ge \sum_{i=1}^n H(D_i \mid D_{< i}) = H(D)
\end{align*}
\else
\[
	H(Z) = \sum_{i=1}^n H(Z_i) = \sum_{i=1}^n H(D_i) \ge \sum_{i=1}^n H(D_i \mid D_{< i}) = H(D)
\]
\fi
and, since $p$ is a product distribution,
\if\conf1
\begin{align*}
H(Z,X) &= - \sum_{i=1}^n \mathrm{E}_{x_i \sim Z_i}[ \log p_i(x_i) ] \\
&= -\sum_{i=1}^n \mathrm{E}_{x_i \sim D_i}[ \log p_i(x_i) ] \\
&= H(D,X).
\end{align*}
\else
\[
	H(Z,X) = - \sum_{i=1}^n \mathrm{E}_{x_i \sim Z_i}[ \log p_i(x_i) ] =
	-\sum_{i=1}^n \mathrm{E}_{x_i \sim D_i}[ \log p_i(x_i) ] = H(D,X).
\]
\fi

Applying \cref{lem:HX-low-cross-entropy} to $X$ and $Z$, we obtain 
\[
	H(X) \le n - \frac{((1-\epsilon)\ell - \frac{1-\epsilon}2 + i)^2}{16(\ell + 4 + 2i)}.
\]
Since the ratio $\frac{((1-\epsilon)\ell - \frac{1-\epsilon}2 + i)^2}{16(\ell + 4 + 2i)}$ is monotone in $i$ when $\ell \ge 1$ and $i \ge 0$, the entropy of $X$ is bounded by
\if\conf1
\begin{align*}
H(X) &\le 
    n - \frac{((1-\epsilon)\ell - \frac{1-\epsilon}2)^2}{16(\ell + 4)} \\
    &= n - \frac{(1-\epsilon)^2}{16} \cdot \frac{(\ell - \frac{1}2)^2}{\ell + 4}.    
\end{align*}
\else
\[
H(X) \le 
    n - \frac{((1-\epsilon)\ell - \frac{1-\epsilon}2)^2}{16(\ell + 4)} = 
    n - \frac{(1-\epsilon)^2}{16} \cdot \frac{(\ell - \frac{1}2)^2}{\ell + 4}.
\]
\fi
Note that when $\ell \ge 1$, then $\frac{(\ell - \frac12)^2}{\ell + 4} \ge \frac{\ell}{20}$. Therefore,
\[
H(X)
    \le n - \frac{(1-\epsilon)^2\ell}{320}. \qedhere
\]
\end{proof}

\subsection{Proof of the List-Decoding Direct Product Theorem}

The proof of \cref{thm:LD-DPT} relies on one more technical observation.

\begin{proposition}
\label{prop:LDDPT-technicality}
For any constant $\lambda > 0$, if $k \in \mathbb{N}$ satisfies
\[
\frac{2^{-\frac{1}{2k}} - 2^{-\frac12}}{1 - 2^{-\frac12}} \ge \lambda
\]
and $\epsilon \in (0,\frac1{2k}]$, then for every $\ell \in \mathbb{N}$,
\[
\frac{2^{-\epsilon \ell} - 2^{-k\epsilon \ell}}{1 - 2^{-k\epsilon \ell}} \ge \lambda^{\ell}.
\]
\end{proposition}

\begin{proof}
As a first step, 
define the function 
\[
g \colon y \mapsto (k-1) - k y + y^k.
\]
This function satisfies $g(1) = 0$.
Its derivative is $g'(y) = -k + k y^{k-1}$, which is negative on all $y \in [0,1)$ and positive on all $y > 1$.
So the function $g$ is non-negative in the entire range $[0,\infty)$. 

Consider now the function $h \colon \epsilon \mapsto \frac{2^{-\epsilon} - 2^{-k\epsilon}}{1 - 2^{-k \epsilon}}$.
The derivative of this function is
\if\conf1
\begin{align*}
    h'(\epsilon) 
    &= \frac{2^{-(k+1)\epsilon} \ln 2}{(1-2^{-k\epsilon})^2} \cdot \big( -(k-1) + k 2^{\epsilon} - (2^\epsilon)^k \big) \\
    &= - C_{k,\epsilon} \cdot g(2^\epsilon)
\end{align*}
\else
\[
h'(\epsilon) = \frac{2^{-(k+1)\epsilon} \ln 2}{(1-2^{-k\epsilon})^2} \cdot \big( -(k-1) + k 2^{\epsilon} - (2^\epsilon)^k \big) = - C_{k,\epsilon} \cdot g(2^\epsilon)
\]
\fi
with $C_{k,\epsilon} = \frac{2^{-(k+1)\epsilon} \ln 2}{(1-2^{-k\epsilon})^2} > 0$ for all $\epsilon > 0$.
By the non-negativity of $g$, we then have that $h$ is non-increasing on $[0,1]$ and so every $\epsilon \in (0, \frac1{2k}]$ satisfies
\[
\frac{2^{-\epsilon} - 2^{-k\epsilon}}{1 - 2^{-k \epsilon}} \ge \frac{2^{-\frac{1}{2k}} - 2^{-\frac12}}{1 - 2^{-\frac12}} \ge \lambda.
\]

To complete the proof, we now want to show that for every $a \in (0,1)$ and every $\ell \ge 1$,
\begin{equation}
\label{eq:ratio-lb}
\frac{a^\ell - a^{k\ell}}{1-a^{k \ell}} \ge \left( \frac{a - a^k}{1-a^k} \right)^\ell.
\end{equation}
Establishing this inequality completes the proof of the proposition by taking $a = 2^{-\epsilon}$ and using the lower bound $\frac{2^{-\epsilon} - 2^{-k \epsilon}}{1-2^{-k \epsilon}} \ge \lambda$ established above.

Fix any $a \in (0,1)$ and $k \in \mathbb{N}$.
Consider the function
\[
\phi \colon \ell \mapsto \ln\left( \frac{a^\ell - a^{k\ell}}{1-a^{k \ell}} \right) - \ell \ln\left( \frac{a - a^k}{1-a^k} \right).
\]
To complete the proof of \cref{eq:ratio-lb}, we want to show that $\phi(\ell) \ge 0$ for all $\ell \in \mathbb{N}$. 
We already have $\phi(1) = 0$.
So we want to show that $\phi'(\ell) \ge 0$ for all $\ell \ge 1$.
The derivative of $\phi$ is
\[
\phi'(\ell) = \ln(a) \left( \frac{a^\ell - k a^{k\ell}}{a^\ell - a^{k\ell}} + \frac{k a^\ell}{1-a^{k\ell}} \right) - \ln \left( \frac{a - a^k}{1 - a^k} \right).
\]

Using the bound $\frac{a - a^k}{1-a^k} < a$ that holds for all $a \in (0,1)$ and $k > 1$, we have that 
\[
\phi'(\ell) \ge  \ln(a) \left( \frac{a^\ell - k a^{k\ell}}{a^\ell - a^{k\ell}} + \frac{k a^\ell}{1-a^{k\ell}} - 1\right).
\]
Since $\ln(a) < 0$, we obtain our proof that $\phi'(\ell) \ge 0$ if we can show that
\[
\frac{a^\ell - k a^{k\ell}}{a^\ell - a^{k\ell}} + \frac{k a^\ell}{1-a^{k\ell}} \le 1.
\]
Expanding, re-arranging the terms, and simplifying, this inequality is equivalent to showing that
\[
0 \le a^k \big( (k-1) - k a + a^k \big) = a^k \cdot g(a).
\]
As we have already seen, $g$ is non-negative and so we also have that $\phi'(\ell) \ge 0$, and as a result \cref{eq:ratio-lb} holds for all $\ell \in \mathbb{N}$, completing the proof of the proposition.
\end{proof}

We are now ready to complete the proof of \cref{thm:LD-DPT}.

\LDDPT*

\begin{proof}
Let us first introduce and fix a few technical parameters.
Set $\lambda = 2^{-\frac1{2560}}$.
Choose $k$ to be the least positive integer that satisfies the inequality $\frac{2^{-1/2k}-2^{-1/2}}{1 - 2^{-1/2}} \ge \lambda$.
Let $\epsilon$ be defined so that $\gamma = 2^{-\epsilon}$.
Fix the universal constant $c$ in the lemma statement to be $c = 2^{-\frac{1}{2k}}$ so that $\epsilon \le \frac{1}{2k}$ when $\gamma \ge c$.

Fix any distribution $\mu$ on the domain of $f$.
Let $A$ be a deterministic algorithm that computes $\LD_{\ell}^n \circ f$ with a success probability of at least $2^{-\epsilon \ell}$ on inputs drawn from $\mu^n$.
Let us define $\beta = 2^{-k\epsilon \ell}$ and $\alpha = \frac{2^{-\epsilon \ell}-\beta}{1-\beta}$.
The algorithm $A$ has probability at least $\alpha$ of landing in a leaf with success probability at least $\beta$ since our choice of $\alpha$ and $\beta$ satisfies
\[
\alpha \cdot 1 + (1-\alpha )\beta  = 2^{-\epsilon \ell}.
\]

By \cref{lem:HX-list}, the entropy loss score of each leaf with success probability at least $\beta = 2^{-k\epsilon \ell}$ is at least $2^{(1-k\epsilon)^2 \ell/320 - n}$.
Therefore, the overall entropy loss score of $A$ is at least
$\alpha \cdot 2^{(1-k\epsilon)^2 \ell/320 - n}$.
By the choice of $k$ and \cref{prop:LDDPT-technicality}, we have $\alpha \ge \lambda^\ell = 2^{-\ell/2560} \ge 2^{-(1-k\epsilon)^2 \ell/640}$.
So the overall entropy score of $A$ is at least $2^{(1-k\epsilon)^2 \ell/640 - n}$.

The argument above implies that
\[
	\R_{2^{-\epsilon \ell}}(\LD_\ell^n \circ f) \ge 
	\R_{\Ent\, 2^{(1-k\epsilon)^2 \ell/640 - n}}^{\mu^n}(f^n).
\]
By the Direct Product Theorem with the entropy loss score, we also have the bound
\if\conf1
\begin{align*}
\R_{\Ent\, 2^{(1-k\epsilon)^2 \ell/640 - n}}^{\mu^n}&(f^n) \ge \\
	&\Omega\left( n \cdot \sR_{\Ent\, 2^{(1-k\epsilon)^2\ell n/640 - 1}}^\mu(f) \right).
\end{align*}
\else
\[
	\R_{\Ent\, 2^{(1-k\epsilon)^2 \ell/640 - n}}^{\mu^n}(f^n) \ge 
	\Omega\left( n \cdot \sR_{\Ent\, 2^{(1-k\epsilon)^2\ell n/640 - 1}}^\mu(f) \right).
\]
\fi
Since the success of any leaf of a randomized algorithm for $f$ is at least $\frac12$,
\[
	\sR_{\Ent\, 2^{(1-k\epsilon)^2\ell n/640 - 1}}^\mu(f) \ge
	\avgR_{\Ent\, 2^{(1-k\epsilon)^2\ell n / 640 - 1}}^\mu(f).
\]

Putting together the inequalities above and choosing $\mu$ that maximizes $\avgR_{\Ent\, 2^{\alpha-1}}^\mu(f)$,
we can then apply \cref{lem:entropy-lin-amplification} to obtain
\if\conf1
\begin{align*}
    \R_{2^{-\epsilon \ell}}(\LD_\ell^n \circ f) 
    &= \Omega \left( (1-k\epsilon)^2 \ell \cdot \R_{\frac23}(f) \right) \\
    &= \Omega \big( \ell \, \R(f) \big). \qedhere
\end{align*}
\else
\[
	\R_{2^{-\epsilon \ell}}(\LD_\ell^n \circ f) = 
	\Omega \left( (1-k\epsilon)^2 \ell \cdot \R_{\frac23}(f) \right) = \Omega \big( \ell \, \R(f) \big). \qedhere
\]
\fi
\end{proof}

\subsection{Proof of the Threshold Direct Product Theorem}

\begin{proof}[Proof of \cref{cor:ThrDPT}]
Let $A$ be an algorithm that computes $\Thr_k^n \circ f$ with success probability at least $\gamma^{\tau_{n,k}}$.
Let $y \in \{0,1\}^n$ denote the output of $A$ on some input $x \in \{0,1\}^{m \times n}$.
Then the set
\[
S_y = \{ z \in \{0,1\}^n : \sum_{i=1}^n \mathbf{1}[z_i \neq y_i] \le n-k \}
\]
includes the true value $z^* = (f(x^{(1)}),\ldots,f(x^{(n)})$ if and only if $A$ correctly computed $\Thr_k^n \circ f$ on its execution. 

Let $B$ be the algorithm that simulates $A$ and outputs the set $S_y$ corresponding to $A$'s output $y$.
The size of the list that $B$ outputs is $|S_y| = 2^{n-\tau_{n,k}}$.
And it computes $\LD_{\tau_{n,k}}^n \circ f$ correctly with  success probabilty $\gamma^{\tau_{n,k}}$.
Therefore, by \cref{thm:LD-DPT}, the algorithm $A$ simulated by $B$ must have query complexity at least $\Omega( \tau_{n,k} \, \R(f))$.
\end{proof}

\subsection{Proof of the Labelled Threshold DPT}

\begin{proof}[Proof of \cref{cor:LabelDPT}]
Let $A$ compute $\mathrm{Label}_k^n \circ f$ with success probability $\gamma^k$. Consider the algorithm $B$ that simulates $A$ and that on $A$'s output $y \in \{0,1,*\}^n$ returns the set
\[
S_y = \big\{ z \in \{0,1\}^n : \forall i \in [n], y_i \in \{z_i,*\}\big\}.
\]
This set has cardinality $|S_y| = 2^{n-k}$ and it includes the true value $(f(x^{(1)}),\ldots,f(x^{(n)}))$ if and only if $y \in \mathrm{Label}_k^n \circ f$.
Therefore, $B$ succeeds in computing $\LD_k^n \circ f$ with probability at least $\gamma^k$ and so by \cref{thm:LD-DPT}, the simulated algorithm $A$ must have query complexity $\Omega( k \, \R(f))$.
\end{proof}

%% file: boosting.tex

\section{Score Amplification}
\label{sec:amplification}

\subsection{Success Amplification by Boosting}

Efficient success amplification is one of the key tools in the study of worst-case and average-case randomized complexity measures.
It is not a property shared by distributional complexity when we fix a single distribution.
(See~\cite{Sha03,BKST24} for example.)
But, interestingly and critically for the proof of the upper bound in \cref{thm:DPT}, the maximum distributional complexity of functions \emph{does} admit efficient success amplification.

\begin{lemma}
For every success probability $\gamma$ in the range $\frac12 \le \gamma \le 1$ and every Boolean-valued function $f \colon \calX \to \{0,1\}$, 
\[
\avgR_{3\gamma^2 - 2\gamma^3}(f) \le 4 \cdot \avgR_{\gamma}(f).
\]
\end{lemma}

\begin{proof}
The proof is a direct extension of Schapire's original boosting algorithm~\cite{Sch90}.

Fix any distribution $\mu$ on $\calX$.
We want to show that $\avgR^\mu_{3 \gamma^2 - 2 \gamma^3} \le 4 \cdot \avgR_\gamma(f)$. 
We will select three randomized algorithms $A_1$, $A_2$, and $A_3$.
For each $i=1,2,3$ and $x \in \calX$, let $\alpha_i(x)$ denote the probability that $A_i$ outputs the corrrect value $f(x)$.

Let $A_1$ be a randomized algorithm with average success probability $\sum_{x \in \calX} \mu(x) \alpha_1(x) = \gamma$ and average cost $\barcost^\mu(A_1) = \avgR^\mu_\gamma(f)$ against $\mu$.

Define $\nu$ to be the distribution on $\calX$ obtained by setting
\[
	\nu(x) = \mu(x) \left( \frac{(1-\gamma)\alpha_1(x) + \gamma(1-\alpha_1(x))}{2\gamma(1-\gamma)} \right)
\]
for every $x \in \calX$. We can verify that $\nu$ is indeed a valid distribution since $\nu(x) \ge 0$ for all $x$ and
\[
\sum_x \nu(x) = \frac{\sum_x \mu(x) \alpha_1(x)}{2\gamma} + \frac{\sum_x \mu(x) (1-\alpha_1(x))}{2(1-\gamma)} = 1.
\]
We select $A_2$ to be a randomized algorithm with average success probability $\sum_{x \in \calX} \nu(x) \alpha_2(x) = \gamma$ and average cost $\barcost^\nu(A_2) = \avgR^\nu_\gamma(f)$ against $\nu$.

Define $\pi$ to be the distribution on $\calX$ obtained by setting
\[
	\pi(x) = \mu(x) \left( \frac{\alpha_1(x)(1-\alpha_2(x)) + (1-\alpha_1(x))\alpha_2(x)}{\Pr_{x \sim \mu}[ A_1(x) \neq A_2(x)]} \right).
\]
We select $A_3$ to be a randomized algorithm with average success probability $\sum_{x \in \calX} \pi(x) \alpha_3(x) = \gamma$ and average cost $\barcost^\pi(A_3) = \avgR^\pi_\gamma(f)$ against $\pi$.

We combine the algorithms $A_1$, $A_2$, and $A_3$ to obtain an algorithm $B$ for $f$ that proceeds as follows. 
On input $x \in \calX$, the algorithm $B$:
\begin{enumerate}
\item Runs $A_1$ on $x$, obtains the value $y_1$.
\item Runs $A_2$ on $x$, obtains the value $y_2$.
\item If $y_1 = y_2$, $B$ outputs the same value.
\item Otherwise, $B$ runs $A_3$ on $x$ and outputs its value.
\end{enumerate}
We now analyze the correctness and the cost of the algorithm $B$.

\paragraph{Correctness analysis.}
We want to show that $B$ has average success probability at least $3\gamma^2 - 2\gamma^3$ when $x$ is drawn from $\mu$.

The probability that $A_1(x) \neq A_2(x)$ and that $A_3(x)= f(x)$ is
\if\conf1
\begin{align*}
\sum_{x \in \calX} &\mu(x) \left( \alpha_1(x)(1-\alpha_2(x)) + (1-\alpha_1(x))\alpha_2(x)\right) \alpha_3(x) \\ 
&= \Pr_{x \sim \mu}[ A_1(x) \neq A_2(x) ] \cdot \sum_{x} \pi(x) \alpha_3(x). 
\end{align*}
\else
\[
	\sum_{x \in \calX} \mu(x) \left( \alpha_1(x)(1-\alpha_2(x)) + (1-\alpha_1(x))\alpha_2(x)\right) \alpha_3(x) = \Pr_{x \sim \mu}[ A_1(x) \neq A_2(x) ] \cdot \sum_{x} \pi(x) \alpha_3(x). 
\]
\fi
By the correctness guarantee on $A_3$, $\sum_x \pi(x) \alpha_3(x) = \gamma$.
And
\if\conf1
\begin{align*}
\Pr_{x \sim \mu}&[ A_1(x) \neq A_2(x) ] \\
&= \sum_{x \in \calX} \mu(x) \left( \alpha_1(x) (1-\alpha_2(x)) + (1-\alpha_1(x))\alpha_2(x) \right).
\end{align*}
\else
\[
	\Pr_{x \sim \mu}[ A_1(x) \neq A_2(x) ] = \sum_{x \in \calX} \mu(x) \left( \alpha_1(x) (1-\alpha_2(x)) + (1-\alpha_1(x))\alpha_2(x) \right).
\]
\fi
The guarantee on $A_1$ implies that $\sum_x \mu(x)\alpha_1(x) = \gamma$ so that 
\if\conf1
\begin{align*}
\Pr_{x \sim \mu}&[ A_1(x) \neq A_2(x) \wedge A_3(x) = f(x)] \\
&= \gamma^2 + \gamma \sum_x \mu(x)\left( (1-\alpha_1(x))\alpha_2(x) - \alpha_1(x) \alpha_2(x)\right).    
\end{align*}
\else
\[
\Pr_{x \sim \mu}[ A_1(x) \neq A_2(x) \wedge A_3(x) = f(x)] = 
	\gamma^2 + \gamma \sum_x \mu(x)\left( (1-\alpha_1(x))\alpha_2(x) - \alpha_1(x) \alpha_2(x)\right).
\]
\fi
Also,
\[
	\Pr_{x \sim \mu}[ A_1(x) = A_2(x) = f(x)] = \sum_{x \in \calX} \mu(x) \alpha_1(x) \alpha_2(x).
\]
So the overall probability that $B$ outputs the correct value $f(x)$ is
\if\conf1
\begin{align*}
\Pr_{x \sim \mu}&[ B(x) = f(x) ] \\
&= \Pr_{x \sim \mu}[ A_1(x) = A_2(x) = f(x)] \\
&\qquad + 
   \Pr_{x \sim \mu}[ A_1(x) \neq A_2(x) \wedge A_3(x) = f(x)] \\
&= \gamma^2 + \sum_x \mu(x) \left( (1-\gamma)\alpha_1(x) + \gamma(1-\alpha_1(x)) \right) \alpha_2(x) \\
&= \gamma^2 + 2\gamma(1-\gamma) \sum_x \nu(x) \alpha_2(x)
\end{align*}
\else
\begin{align*}
\Pr_{x \sim \mu}[ B(x) = f(x) ]
&= \Pr_{x \sim \mu}[ A_1(x) = A_2(x) = f(x)] + 
   \Pr_{x \sim \mu}[ A_1(x) \neq A_2(x) \wedge A_3(x) = f(x)] \\
&= \gamma^2 + \sum_x \mu(x) \left( (1-\gamma)\alpha_1(x) + \gamma(1-\alpha_1(x)) \right) \alpha_2(x) \\
&= \gamma^2 + 2\gamma(1-\gamma) \sum_x \nu(x) \alpha_2(x)
\end{align*}
\fi
and the correctness guarantee on $A_2$ implies that $\sum_x \nu(x) \alpha_2(x) = \gamma$ which means that the overall success probability is $\gamma^2 + 2\gamma^2(1-\gamma) = 3 \gamma^2 - 2 \gamma^3$, as we wanted to show.

\paragraph{Cost analysis.}
The average cost of $A_1$ against $\mu$ is bounded by
\[
\barcost^\mu(A_1) = \avgR^\mu_\gamma(f) \le \avgR_\gamma(f).
\]
Let $c_2(x)$ denote the average cost of $A_2$ on input $x$.
The average cost of $A_2$ with respect to $\mu$ is
\if\conf1
\begin{align*}
\barcost^\mu(A_2) &= \sum_x \mu(x) c_2(x) \\
&= \sum_x \frac{2}{\alpha_1(x)/\gamma + (1-\alpha_1(x))/(1-\gamma)} \nu(x) c_2(x).
\end{align*}
\else
\[
\barcost^\mu(A_2) = \sum_x \mu(x) c_2(x) = \sum_x \frac{2}{\alpha_1(x)/\gamma + (1-\alpha_1(x))/(1-\gamma)} \nu(x) c_2(x).
\]
\fi
For any input $x$, $0 \le \alpha_1(x) \le 1$ and so
\[
\frac{2}{\alpha_1(x)/\gamma + (1-\alpha_1(x))/(1-\gamma)}
\le 2 \max\{\gamma, 1-\gamma\} \le 2.
\]
Also, the cost guarantee on $A_2$ implies that $\sum_x \nu(x) c_2(x) \le \avgR_\gamma(f)$ so
\[
\barcost^\mu(A_2) \le 2 \sum_x \nu(x) c_2(x) \le 2 \avgR_\gamma(f).
\]

Finally, for every $x \in \calX$, the algorithm $B$ runs $A_3$ with probability $\alpha_1(x)(1-\alpha_2(x)) + (1-\alpha_1(x))\alpha_2(x)$.
So if we let $c_3(x)$ denote the expected cost of $A_3$ on input $x$, the average contribution of $A_3$ to the overall cost of $B$ is
\if\conf1
\begin{align*}
\sum_x &\mu(x) \big( \alpha_1(x)(1-\alpha_2(x)) + (1-\alpha_1(x))\alpha_2(x) \big) c_3(x)
\\ &= \Pr_{x \sim \mu}[ A_1(x) \neq A_2(x)] \sum_x \pi(x) c_3(x).
\end{align*}
\else
\[
\sum_x \mu(x) \big( \alpha_1(x)(1-\alpha_2(x)) + (1-\alpha_1(x))\alpha_2(x) \big) c_3(x)
= \Pr_{x \sim \mu}[ A_1(x) \neq A_2(x)] \sum_x \pi(x) c_3(x).
\]
\fi
The cost guarantee on $A_3$ implies that $\sum_x \pi(x) c_3(x) \le \avgR_\gamma(f)$.
So its overall contribution to the cost of $B$ is at most 
$\Pr_{x \sim \mu}[ A_1(x) \neq A_2(x)] \cdot \avgR_\gamma(f) \le \avgR_\gamma(f)$.

In total, the average cost of $B$ over $\mu$ is thus bounded above by
\[
\barcost^\mu(B) \le \avgR_\gamma(f) + 2 \avgR_\gamma(f) + \avgR_\gamma(f) = 4 \avgR_\gamma(f). \qedhere
\]
\end{proof}

The boosting lemma immediately yields the score amplification result stated in the preliminaries, which we restate here for convenience.

\boosting*

\begin{proof}
This lemma follows immediately from the last one by observing that the difference in the success probability $3\gamma^2 - 2\gamma^3$ after boosting and the original success probability $\gamma$ is
\[
	3\gamma^2 - 2\gamma^3 - 1 = (2\gamma - 1)(1-\gamma).
\]
When $\frac{1+\delta}2 \le \gamma \le \frac23$, this difference is bounded below by $(2\gamma - 1)(1-\gamma) \ge \big(2(\tfrac{1+\delta}2) - 1\big)\big(1 - \tfrac23\big) = \frac{\delta}3$.
And when $\frac23 \le \gamma \le 1-\delta$, we have $(2\gamma - 1)(1-\gamma) \ge \big(2(\tfrac23)-1\big)\big(1-(1-\delta)\big) = \frac{\delta}3$.
\end{proof}

\subsection{Hellinger Score Amplification}

The success amplification result from the last section also shows that for every small $\delta$, $\R_{\frac23}(f) = \Theta(\avgR_{\frac23}(f)) = 
\Theta( \frac1{\delta^2} \avgR_{\frac{1+\delta}2}(f))$.
When we use Hellinger score, we can get more efficient success amplification where the multiplicative factor is linear instead of quadratic in $\frac1\delta$.

The technical lemma at the heart of the proof is stated more naturally in terms of \emph{Hellinger loss} instead of Hellinger score.
Define the Hellinger loss of a leaf $\ell$ with respect to $f$ and $\mu$ to be $\HelLoss_f^\mu(\ell) = 2\Hel_f^\mu(\ell) - 1$.
The Hellinger loss of a randomized algorithm is the expected Hellinger loss of its leaves:
$\barHelLoss_f^\mu(R) = \E_{D \sim R} \E_{\ell \sim D(\mu)} \HelLoss_f^\mu(\ell)$.
We write $\avgR_{\HelLoss\, \epsilon}(f)$ to denote the maximum distributional query complexity of $f$ for the set of algorithms with Hellinger loss at most $\epsilon$.

\begin{lemma}
\label{lem:hellinger-boosting}
For any $\alpha \in [\frac12,1]$, $\beta \in [0,1]$ and function $f \colon \{0,1\}^m \to \{0,1\}$,
\[
	\avgR_{\HelLoss\,\alpha(1 - \frac{1-\beta}6)}(f) \le \avgR_{\HelLoss\,\alpha}(f) + 2 \,\avgR_{\HelLoss\,\beta}(f).
\]
\end{lemma}

\begin{proof}
Fix any distribution $\mu$ on the inputs to $f$.
Take $A$ to be a randomized decision tree with Hellinger loss at most $\alpha$ and average cost at most $\barR_{\HelLoss\,\alpha}(f)$ with respect to $\mu$.

For any set leaf $\ell$ of a decision tree in the support of $A$, define the distribution
\[
\nu_\ell(x) = \begin{cases}
	\frac{\mu(x)}{2\mu(\ell_0)} & \mbox{if } x \in \ell_0 \\
	\frac{\mu(x)}{2\mu(\ell_1)} & \mbox{if } x \in \ell_1 \\
	0 & \mbox{if } x \notin \ell.
\end{cases}
\]
Let $B_\ell$ denote a randomized algorithm with Hellinger loss at most $\beta$ and cost at most $\avgR_{\HelLoss\,\beta}(f)$ against the distribution $\nu_\ell$.

Consider now the randomized algorithm $C$ defined by the following procedure:
First, we run $A$ on the input $x$.
Let $\ell$ denote the leaf of the tree reached by input $x$ and the subset of inputs $\{0,1\}^m$ that lead to this leaf.
Then we run $B_\ell$ on $x$.

The Hellinger loss of the algorithm $C$ is bounded by
\if\conf1
\begin{align*}
\barHelLoss_f^\mu(C) 
&= \E_{D \sim A} \sum_{\ell \in D} \E_{D' \sim B_\ell} \sum_{\ell' \in D'} 2\sqrt{\mu(\ell'_0) \mu(\ell'_1)} \\
&= \E_{D \sim A} \sum_{\ell \in D} \E_{D' \sim B_\ell} \sum_{\ell' \in D'} 2\sqrt{\mu(\ell_0) \mu(\ell_1)} \\
&\qquad \qquad \qquad \qquad \qquad \cdot 2\sqrt{\nu_\ell(\ell'_0) \nu_\ell(\ell'_1)} \\
&= \E_{D \sim A} \sum_{\ell \in D} 2\sqrt{\mu(\ell_0) \mu(\ell_1)} \\
&\qquad \qquad \cdot \left( \E_{D' \sim B_\ell} \sum_{\ell' \in D'} 2\sqrt{\nu_\ell(\ell'_0) \nu_\ell(\ell'_1)} \right) \\
&\le \alpha \cdot \beta.
\end{align*}
\else
\begin{align*}
\barHelLoss_f^\mu(C) 
&= \E_{D \sim A} \sum_{\ell \in D} \E_{D' \sim B_\ell} \sum_{\ell' \in D'} 2\sqrt{\mu(\ell'_0) \mu(\ell'_1)} \\
&= \E_{D \sim A} \sum_{\ell \in D} \E_{D' \sim B_\ell} \sum_{\ell' \in D'} 2\sqrt{\mu(\ell_0) \mu(\ell_1)} \cdot 2\sqrt{\nu_\ell(\ell'_0) \nu_\ell(\ell'_1)} \\
&= \E_{D \sim A} \sum_{\ell \in D} 2\sqrt{\mu(\ell_0) \mu(\ell_1)} \cdot \left( \E_{D' \sim B_\ell} \sum_{\ell' \in D'} 2\sqrt{\nu_\ell(\ell'_0) \nu_\ell(\ell'_1)} \right) \\
&\le \alpha \cdot \beta.
\end{align*}
\fi

By construction, the expected cost of $A$ with respect to $\mu$ is at most $\avgR_{\HelLoss\,\alpha}(f)$.
And for any set $\ell$, the expected cost of $B_\ell$ on inputs drawn from $\mu$ conditioned on being in $\ell$ is
\if\conf1
\begin{align*}
\sum_{x \in \ell} &\frac{\mu(x)}{\mu(\ell)} \barcost(B_\ell,x) \\
&= \sum_{x \in \ell_0} \frac{2\mu(\ell_0)}{\mu(\ell)} \nu_\ell(x) \barcost(B_\ell,x) \\
&\quad + 
  \sum_{x \in \ell_1} \frac{2\mu(\ell_1)}{\mu(\ell)} \nu_\ell(x) \barcost(B_\ell,x) \\
&\le 2 \frac{\max\{\mu(\ell_0),\mu(\ell_1)\}}{\mu(\ell)}  \sum_x \nu_\ell(x) \barcost(B_\ell,x).
\end{align*}
\else
\begin{align*}
\sum_{x \in \ell} \frac{\mu(x)}{\mu(\ell)} \barcost(B_\ell,x) 
&= \sum_{x \in \ell_0} \frac{2\mu(\ell_0)}{\mu(\ell)} \nu_\ell(x) \barcost(B_\ell,x) + 
  \sum_{x \in \ell_1} \frac{2\mu(\ell_1)}{\mu(\ell)} \nu_\ell(x) \barcost(B_\ell,x) \\
&\le 2 \frac{\max\{\mu(\ell_0),\mu(\ell_1)\}}{\mu(\ell)} \sum_x \nu_\ell(x) \barcost(B_\ell,x).
\end{align*}
\fi
By the trivial upper bound $\max\{\mu(\ell_0),\mu(\ell_1)\} \le \mu(\ell)$ and the guarantee on the expected cost of $B_\ell$ under the distribution $\nu_\ell$, we then have that 
$\overline{\cost}(B_\ell,\mu|_\ell) \le 2\, \avgR_\beta(f)$.
Therefore, the overall expected cost of $B$ is at most $\avgR_{\HelLoss\,\alpha}(f) + 2 \, \avgR_{\HelLoss\,\beta}(f)$, as claimed.
\end{proof}

We then get the linear amplification lemma stated in the preliminaries as a corollary.

\linamplification*

\begin{proof}
By repeated application of \cref{lem:hellinger-boosting} with the initial parameter $\alpha=1$ and the observation that $\avgR_{\HelLoss\,1}(f) = 0$ for all functions $f$, we have that for any $k \ge  1$ and $\beta \in [0,1]$,
\[
	\avgR_{\HelLoss\,\beta^k}(f) \le 2k \cdot \avgR_{\HelLoss\,\beta}(f).
\]
Therefore, taking $k = \lceil \frac{2}{1-\beta} \rceil$, we have that 
\[
\avgR_{\HelLoss\,e^{-1}}(f) = O \left( \frac{1}{1-\beta} \avgR_{\HelLoss\,\beta}(f) \right).
\]
Since Hellinger score $\frac{1+\delta}2$ is equivalent to Hellinger loss $1-\delta$, the above identity implies that
\[
\avgR_{\HelLoss\,e^{-1}}(f) = O \left( \frac{1}{\delta} \cdot \avgR_{\Hel\,\frac{1+\delta}2}(f) \right).
\]
The lemma then follows from observing that $\R(f) = O(\avgR_{\HelLoss\,c}(f))$ for any constant $c$ bounded away from $0$ and $1$.
\end{proof}

The linear amplification lemma for Hellinger score also immediately implies an analogous lemma for exponential entropy score.

\begin{lemma}
\label{lem:entropy-lin-amplification}
For any $\alpha$, $0 \le \alpha \le 1$,
\[
	\R(f) = O \left( \frac{1}{\alpha} \cdot \avgR_{\Ent\,2^{\alpha-1}}(f) \right).
\]
\end{lemma}

\begin{proof}
For any $x \in [0,1]$, the convexity of the function $2^{-x}$ implies that $2^{-x} \le 1 - \frac x2$.
As a result, the exponential entropy score of a leaf with density $q$ is
\[
q^q(1-q)^q = 2^{-h(q)} \le 1 - \frac12 h(q) = \frac{1 + (1-h(q))}{2}.
\]
The function $1-h(q)$ is in turn bounded above by 
\[
1-h(q) \le \frac1{\ln 2} (1 - 2 \sqrt{q(1-q)}).
\]
Consider any randomized algorithm $R$ with exponential entropy score at least $2^{\alpha - 1}$ with respect to $f$ and a fixed distribution $\mu$.
Then with the two inequalities above,
\if\conf1
\begin{align*}
2^{\alpha - 1} &\le \barEnt^\mu_f(R) \\
&= \E_{\ell} 2^{-h(\mu(\ell_1)/\mu(\ell))} \\
&\le \frac{1 + (\ln 2)^{-1}\E_\ell\big[1 - 2\sqrt{\frac{\mu(\ell_1)}{\mu(\ell)}\cdot \frac{\mu(\ell_0)}{\mu(\ell)}}\big]}{2}.
\end{align*}
\else
\[
2^{\alpha - 1} \le \barEnt^\mu_f(R) = \E_{\ell} 2^{-h(\mu(\ell_1)/\mu(\ell))}
\le \frac{1 + (\ln 2)^{-1}\E_\ell\big[1 - 2\sqrt{\frac{\mu(\ell_1)}{\mu(\ell)}\cdot \frac{\mu(\ell_0)}{\mu(\ell)}}\big]}{2}.
\]
\fi
This inequality and the basic inequality $1 + \ln(2) \alpha \le 2^\alpha$ imply
\[
\ln(2)^2 \alpha \le \E_\ell\left[1 - 2\sqrt{\frac{\mu(\ell_1)}{\mu(\ell)}\cdot \frac{\mu(\ell_0)}{\mu(\ell)}}\right].
\]
And the Hellinger score of $R$ is therefore bounded below by
\if\conf1
\begin{align*}
\barHel^\mu_f(R) &= \E_\ell\left[1 - \sqrt{\frac{\mu(\ell_1)}{\mu(\ell)}\cdot \frac{\mu(\ell_0)}{\mu(\ell)}}\right] \\
&= \frac{1 + \E_\ell\left[1 - 2\sqrt{\frac{\mu(\ell_1)}{\mu(\ell)}\cdot \frac{\mu(\ell_0)}{\mu(\ell)}}\right]}{2} \\
&\ge \frac{1 + \ln(2)^2 \alpha}{2}.
\end{align*}
\else
\[
\barHel^\mu_f(R) = \E_\ell\left[1 - \sqrt{\frac{\mu(\ell_1)}{\mu(\ell)}\cdot \frac{\mu(\ell_0)}{\mu(\ell)}}\right]
= \frac{1 + \E_\ell\left[1 - 2\sqrt{\frac{\mu(\ell_1)}{\mu(\ell)}\cdot \frac{\mu(\ell_0)}{\mu(\ell)}}\right]}{2} \ge \frac{1 + \ln(2)^2 \alpha}{2}.
\]
\fi

We have just shown that for every distribution $\mu$, every randomized algorithm $R$ with entropy score $\barEnt^\mu_f(R) \ge 2^{\alpha - 1}$ also has Hellinger score $\barHel^\mu_f(R) \ge \frac{1 + \ln(2)^2 \alpha}{2}$. As a result, $\avgR_{\Ent\,2^{\alpha-1}}(f) \le
 \avgR_{\Hel\,\frac{1+ \ln(2)^2 \alpha}2}(f)$ and the lemma is an immediate consequence of \cref{lem:hellinger-lin-amp}.
\end{proof}

%% file: relations.tex

\section{Relations Between the Different Complexity Measures}

We prove various relations between different complexity measures in this section.

\subsection{Worst-Case and Average-Case Complexity}
\label{sec:wc-ac}

Relations between worst-case and average-case complexity measures are well known. We include them here for the convenience of the reader.

As a starting point, by definition, the worst-case and average-case randomized query complexity measures always satisfy the inequality
\[
	\barR_\gamma(f) \le \R_\gamma(f).
\]
In the other direction, we have the following inequalities.

\begin{proposition}
\label{prop:wc-ac}
For every $\frac12 < \gamma < 1$ and every function $f$,
\[
	\R_{\gamma}(f) \le O\left( \frac{1}{\min\{\gamma - \frac12, 1 - \gamma\}} \cdot \barR_\gamma(f) \right).
\]
\end{proposition}

\begin{proof}
Consider first the small-advantage regime where $\gamma = \frac{1+\delta}{2}$. 
Let $A$ be a randomized algorithm with success probability $\gamma$ and average cost $\barR_\gamma(f)$.
Let $A'$ be the algorithm obtained by truncating $A$ and forcing it to guess the value of $f$ whenever $A$ exceeds its expected cost by a multiplicative factor of $4/\delta$.
By Markov's inequality, the probability that this truncation event occurs is at most $\delta/4$, so that the success probability of $A'$ is at least $\frac{1 + \delta/2}{2}$ and we obtain the desired bound on the worst-case complexity for success probability $\frac{1+\delta}{2}$ by applying standard success amplification.

For the small-error regime where $\gamma = 1 - \epsilon$, we use the same argument with truncation factor $1/\epsilon$.
\end{proof}

\Cref{prop:wc-ac} implies that the average-case and worst-case randomized query complexity measures are asymptotically equivalent in the bounded-error regime, when $\gamma$ is bounded away from both $\frac12$ and $1$ by a constant.
In the small-bias and small-error regimes, however, there are functions whose average-case randomized query complexity is asymptotically smaller than their worst-case randomized query complexity.
For the small-bias regime, the bound in the last proposition is optimal.

\begin{proposition}
For every $n \in \mathbb{N}$, there exists a function $f \colon \{0,1\}^n \to \{0,1\}$ such that for every $\delta = \delta(n) > 0$
\[
	\barR_{\frac{1+\delta}2}(f) \le \delta \cdot \R_{\frac{1+\delta}2}(f).
\]
\end{proposition}

\begin{proof}
Consider the parity function $\oplus_n \colon x \mapsto \sum_{i=1}^n x_i \pmod{2}$.
For any $\delta > 0$, $\R_{\frac{1+\delta}2}(\oplus_n) = n$.
But in the average-case complexity setting, the (lazy) algorithm that with probability $\delta$ queries everything to compute the value of $\oplus_n(x)$ and otherwise guesses the value of the function has average cost $\delta n = \delta \, \R(\oplus_n)$ and succeeds with probability $\delta + (1-\delta)\frac12 = \frac{1+\delta}2$ on every input. 
\end{proof}

Note that by contrast there are also other functions for which $\barR_{\frac{1+\delta}2}(f) = \Theta(\R_{\frac{1+\delta}2}(f))$. 
One example of such a function is the Majority function when $\delta = 1/n$.

There can also be asymptotic separations between the average-case and worst-case randomized query complexities in the small-error regime.

\begin{proposition}
For infinitely many values of $n \in \mathbb{N}$, there exists a partial function $f \colon X \to \{0,1\}$, $X \subseteq \{0,1\}^n$, such that for every $\epsilon = \epsilon(n) > 0$,
\[
	\barR_{1-\epsilon}(f) \le O \left( \frac{1}{\log(1/\epsilon)} \cdot \R_{1-\epsilon}(f) \right).
\]
\end{proposition}

\begin{proof}
Fix any even value of $n$ and define the set
\if\conf1
\begin{align*}
X = \Big\{x& \in \{0,1\}^n : \sum_{i=1}^{n/2} x_{2i-1} = \tfrac{n}2 \ \wedge \\
&\exists b \in \{0,1\} \forall i \in [\tfrac n2] \ x_{2i-1} = 1 \Rightarrow x_{2i} = b\Big\}. 
\end{align*}
\else
\[
X = \left\{x \in \{0,1\}^n : \sum_{i=1}^{n/2} x_{2i-1} = \tfrac{n}2 \wedge \exists b \in \{0,1\} \forall i \in [\tfrac n2] \ x_{2i-1} = 1 \Rightarrow x_{2i} = b\right\}.
\]
\fi
For each $x \in X$, we define $f(x)$ to be the value of $b$ for which the last condition is satisfied.

Informally, we can think of $f$ as containing $\frac n2$ bins. Half of the bins are useful in that the second bit of the bin determines the value of the function and the other half are useless. Furthermore, the first bit of each bit indicates whether the bin is useful or not.

An algorithm that computes $f$ with error at most $\epsilon$ must find a useful bin with probability at least $1-O(\epsilon)$.
This can be guaranteed only when the worst-case query complexity of the algorithm is $\Omega( \log \frac1{\epsilon})$.
So $\R_{1-\epsilon}(f) = \Omega( \log \frac1{\epsilon})$.

But the average-case cost of the algorithm that examines bins uniformly at random follows a geometric distribution with success parameter $\frac12$ and so $\avgR_{1-\epsilon}(f) = O(1)$.
\end{proof}

There are also total functions that exhibit an analogous gap between the worst-case and average-case query complexity; see~\cite{BB19}.

\subsection{Maximum Distributional Complexity}
\label{sec:max-dist-relations}

Any algorithm with score $\gamma$ and average cost at most $c$ on every input $x$ also has the same score and cost guarantees over all distributions $\mu$ on the inputs.
So for every $f$ and $\gamma$,
\[
	\avgR_\gamma(f) \le \barR_\gamma(f).
\]

This inequality is asymptotically optimal in the bounded-error regime.
It is also asymptotically optimal in the small-error regime, where Vereshchagin~\cite{Ver98} showed that 
$\barR_{1-\epsilon}(f) \le 2 \avgR_{1-\epsilon/2}(f)$. 

In the small-advantage regime, we have the following relation between the maximum distributional and average-case complexity.

\begin{proposition}
\label{prop:avgR-barR}
For every function $f \colon \{0,1\}^n \to \{0,1\}$ and $\delta = \delta(n) > 0$, the average-case and maximum distributional complexity measures for success probability $\frac{1+\delta}2$ are related by
\[
	\barR_{\frac{1+\delta}2}(f) \le O \left( \frac1{\delta} \cdot \avgR_{\frac{1+\delta}2}(f) \right).
\]
\end{proposition}

\begin{proof}
We adapt Vereshchagin's argument~\cite{Ver98} for the result mentioned above to the small-advantage regime.

Fix any pair of distributions $\mu$ and $\nu$.
For every $0 \le \lambda \le 1$, let $\pi_\lambda = \lambda \mu + (1-\lambda) \nu$.
And let $A_\lambda$ be a randomized algorithm with success at leat $\frac{1+\delta}2$ and average cost at most $\avgR_{\frac{1+\delta}2}(f)$ when both are measured against $\mu$.
The success probability of $A_\lambda$ satisfies
\[
	\barsuccess_f^{\pi_\lambda}(A_\lambda) = \lambda \cdot \barsuccess_f^{\mu}(A_\lambda) + (1-\lambda) \barsuccess_f^{\nu}(A_\lambda).
\]
Since $\barsuccess_f^\mu(A_\lambda) \le 1$, the success of $A_\lambda$ against $\nu$ is bounded below by
\[
	\barsuccess_f^\nu(A_\lambda) \ge \frac{(1+\delta)/2 - \lambda}{1-\lambda}.
\]
Similarly, the average cost of $A_\lambda$ satisfies
\[
	\barcost^{\pi_\lambda}(A_\lambda) = \lambda \barcost^\mu(A_\lambda) + (1-\lambda) \barcost^\nu(A_\lambda)
\]
so that
\[
	\barcost^\mu(A_\lambda) \le \frac1\lambda \cdot \avgR_{\frac{1+\delta}2}(f).
\]
By choosing $\lambda = \delta/4$, we obtain an algorithm $A$ with success probability $\barsuccess_f^\nu(A) \ge \frac{1+\delta/2}{2}$ and average cost $\barcost^\mu(A) \le \frac4\delta \avgR_{\frac{1+\delta}2}(f)$.
This means that 
\[
	\barR_{\frac{1+\delta/2}2}(f) \le \frac4\delta \cdot \avgR_{\frac{1+\delta}2}(f)
\]
and the result now follows from success amplification for average-case complexity.
\end{proof}

The bound in \cref{prop:avgR-barR} is qualitatively tight in that there can be a polynomial factor of $\delta$ that separates the average-case and maximum distributional complexity of functions.

\begin{proposition}
For infinitely many $n \in \mathbb{N}$, there exists a partial function $f \colon X \to \{0,1\}$, $X \subseteq \{0,1\}^n$ such that for all $\delta = \delta(n) > 0$,
\[
	\avgR_{\frac{1+\delta}2}(f) = O\left( \delta^{\frac16} \,\barR_{\frac{1+\delta}2}(f) \right).
\]
\end{proposition}

\begin{proof}
Consider the following construction in which we split the $n$ coordinates into $k$ blocks $B_1,\ldots,B_k$ of $n/k$ variables each. 
We construct the function so that there are two types of blocks: hard blocks and easy blocks.
The first variable in each block determines the type of a block, so that the algorithm can tell if a block is hard or easy with a single query.
We also associate a value with each block.
For each easy block, the value is determined by the second variable in the block, so that the value of an easy block is determined with only one additional query.
The value of a hard block, by contrast, is the parity of $1/\delta$ variables, so that this value can only be determined by making $1/\delta$ queries to the block.

Formally, writing $x \in \{0,1\}^n$ as $x = (x^{(1)},\ldots,x^{(k)})$ with each $x^{(i)} \in \{0,1\}^{n/k}$, we define the function $g \colon \{0,1\}^{n/k} \to \{0,1\}$ by
\[
g(y) = \begin{cases}
    y_2 & \mbox{if } y_1 = 0 \\
    y_2 \oplus \cdots \oplus y_{1/\delta+1} & \mbox{if } y_1 = 1.
\end{cases}
\]
Then the function $f \colon \{0,1\}^n \to \{0,1\}$ is
\[
f(x) = \mathrm{Maj}(g(x^{(1)}),\ldots,g(x^{(k)})).
\]
Let us consider the partial function obtained by restricting $f$'s domain $X$ to one of two possible types of input:
\begin{description}
\item[Many hard blocks:] 
\if\conf1 \phantom{Hack} \\ \fi
For inputs of this type, there are exactly $\sqrt{\delta} k$ blocks that are hard. And their values are balanced: exactly half of them have value 0 and the other half have value 1. The values of the easy blocks have $\delta$ bias towards the value of the function. (I.e., for inputs of this type, the function $f$ is the Gap Majority function on the values of the easy blocks.)

\item[Few hard blocks:] 
\if\conf1 \phantom{Hack} \\ \fi
For inputs of this type, exactly $\delta k$ blocks are hard. Each of these blocks has the same value, whereas the values of the easy blocks are balanced. As a result, the value of the function is the value of each of the hard blocks.
\end{description}

We now prove the two complexity bounds separately.
First, we establish the upper bound 
\[
\avgR_{\frac{1+\delta}2}(f) = O(1/\delta^3)
\]
on the maximum distributional complexity of $f$.
We do so by considering two natural algorithms for computing $f$.
Algorithm $A$ picks a random block. If the block is easy, it outputs its value; otherwise it just guesses.
Algorithm $B$ picks a random block. If the block is hard, it computes its value and outputs it; otherwise it just guesses.

Fix now any distribution $\mu$ on the domain of the partial function.
Let $p$ denote the probability that an input drawn from $\mu$ has many hard blocks.

Consider first the case where $p \ge \delta^{\frac16}$. 
Algorithm $A$ has success probability at least $\frac{1+ p \delta}2$ on inputs drawn from $\mu$. By extending the algorithm so that it queries the value of $O(1/p^2)$ easy blocks instead of a single one (ignoring hard blocks, as before), we obtain an algorithm $A'$ with success probability $\frac{1+\delta}2$ and average cost $O(1/p^2) \le O(\delta^{-1/3})$ with respect to $\mu$.

Consider now the case where $p < \delta^{\frac16}$. 
Algorithm $B$ has success probability at least $\frac{1+(1-p)\delta}2$ on inputs drawn from $\mu$. By considering the extension of this algorithm that queries two blocks to see if either of them is hard, and outputting the value of a hard block if one is found, we obtain algorithm $B'$ with success probability at least $\frac{1+(1-p)(2\delta - \delta^2)}{2}$ against $\mu$, which is at least $\frac{1+\delta}2$ when $p \le \delta^{\frac16}$. And the average cost of $B'$ over $\mu$ is 
\if\conf1
\begin{align*}
2 + p (2\sqrt{\delta}-\delta) (1/\delta) + &(1-p)(2\delta - \delta^2)(1/\delta) \\
&\le 2 + 2p/\sqrt{\delta} + 2 
= O(p/\sqrt{\delta}),
\end{align*}
\else
\[
2 + p (2\sqrt{\delta}-\delta) (1/\delta) + (1-p)(2\delta - \delta^2)(1/\delta) \le 2 + 2p/\sqrt{\delta} + 2 = O(p/\sqrt{\delta}),
\]
\fi
which is $O(\delta^{-1/3})$ when $p < \delta^{\frac16}$.

Combining the two observations above, we obtain that for any distribution $\mu$ on domain $X$ of $f$, there is an algorithm (either $A'$ or $B'$) with success probability $\frac{1+\delta}2$ against $\mu$ and average cost $O(\delta^{-1/3})$ with respect to the same distribution.
Therefore, by definition $O(\delta^{-1/3})$ is an upper bound on $\avgR_{\frac{1+\delta}2}(f)$.

To complete the proof of the proposition, we now show that the average-case complexity of $f$ is bounded below by
\[
\barR_{\frac{1+\delta}2}(f) = \Omega( \delta^{-1/2} ).
\]
To do so, assume on the contrary that there is an algorithm that computes $f$ with success probability at least $\frac{1+\delta}2$ and average cost $q = o(\sqrt{1/\delta})$.
Any algorithm with this cost can only check $o(\sqrt{1/\delta})$ blocks to see if they are hard or easy and it can compute the value of at most that number of easy blocks.
The goal of an algorithm seeking to minimize its average cost is to avoid computing the value of any hard block when there are many of them.

Without loss of generality, we can consider algorithms that decide whether to compute a hard block based on (i) the number of hard blocks observed; and (ii) the bias of the easy blocks that is computed.
Note that when there are many hard blocks, the easy blocks have bias $\delta$ towards the function value, whereas when there are few hard blocks, the easy blocks's values are balanced. However, with at most $o(1/\sqrt{\delta})$ easy block values, an algorithm cannot distinguish between these two cases with bias more than $o(1)$. Therefore, without loss of generality we can consider algorithms that use only the number of hard blocks observed to determine whether they should compute the value of one of those blocks.

When we run an algorithm with query complexity $o(\sqrt{1/\delta})$ on inputs that have few hard blocks,
the probability that the algorithm finds at least 1 block is at most $q \cdot \delta = o(\sqrt{\delta})$. 
And the probability that it finds more than one hard block is at most $\binom{q}{2} \delta^2 = o(\delta)$. 
So to have success probability at least $\frac{1+\delta}2$, its probability of computing a hard block (over all possible biases of the easy blocks) when there is exactly one such block must be at least $\sqrt{\delta}$.
But then the probability that the algorithm sees exactly 1 hard block on inputs that have many hard blocks is constant. Which means that its average cost on the uniform distribution over those inputs is at least $\sqrt{\delta} \cdot 1/\delta = 1/\delta^2$.
\end{proof}

We remark that the last proposition also implies a polynomial separation between average-case and maximum distributional complexity: there exists a partial function $f$ and success parameter $\gamma$ for which $\avgR_\gamma(f) = O \left( \barR_\gamma(f)^{2/3} \right)$.
We have not tried to determine if the same is true for total function (though we suspect that it is) and what is the largest polynomial separation that we could achieve between the two complexity measures.

\subsection{Score-Weighted Complexity}
\label{sec:relations-sR}

\subsubsection{Relation to Maximum Distributional Complexity}

We can give an upper bound as well as a lower bound of the score-weighted query complexity of randomized algorithms in terms of their maximum distributional complexity. 

\sRandAvgR*

\begin{proof}
For any distribution $\mu$, we can express the score-weighted cost of $f$ as 
\[
	\sR_\gamma^\mu(f) = \min_{R \,:\, \barsfmu(R) \ge \gamma} \frac{\E_{\ell \sim L(\mu,R)}[ \barsfmu(\ell) \cdot \cost(\ell) ]}{\barsfmu(R)}
\]
where $L(\mu,R)$ the distribution on the leaves $\ell = D_l(x)$ obtained by drawing $D \sim R$ and $x \sim \mu$.

Since scoring measures are bounded above by $1$,
\if\conf1
\begin{align*}
\sR_\gamma(f) 
 &= \max_\mu \sR_\gamma^\mu(f) \\
 &\le \max_\mu \min_{R \,:\, \barsfmu(R) \ge \gamma} \frac{\E_{\ell}[ 1 \cdot \cost(\ell) ]}{\gamma} \\
 &= \frac1\gamma \cdot \avgR_\gamma(f).
\end{align*}
\else
\[
	\sR_\gamma(f) 
 = \max_\mu \sR_\gamma^\mu(f)
 \le \max_\mu \min_{R \,:\, \barsfmu(R) \ge \gamma} \frac{\E_{\ell}[ 1 \cdot \cost(\ell) ]}{\gamma} = \frac1\gamma \cdot \avgR_\gamma(f).
\]
\fi
And since they are also bounded below by $\frac12$,
\if\conf1
\begin{align*}
\avgR_\gamma(f)
&= \max_\mu \avgR_\gamma^\mu(f) \\
&= 2 \cdot \max_\mu \min_{R \,:\, \barsfmu(R) \ge \gamma} \frac{\E_{\ell}[ \frac12 \cdot \cost(\ell) ]}{1} \\
&\le 2 \cdot \sR_\gamma(f). \qedhere
\end{align*}
\else
\[
\avgR_\gamma(f)
= \max_\mu \avgR_\gamma^\mu(f)
= 2 \cdot \max_\mu \min_{R \,:\, \barsfmu(R) \ge \gamma} \frac{\E_{\ell}[ \frac12 \cdot \cost(\ell) ]}{1}
\le 2 \cdot \sR_\gamma(f). \qedhere
\]
\fi
\end{proof}

\subsubsection{Relation to Worst-Case Query Complexity}

We can also upper bound score-weighted complexity with worst-case complexity for every function and for their direct products.

\sRandR*

\begin{proof}
Fix any distribution $\mu$ on the domain of $f^n$. 
Let $\calR_\mu$ denote the set of randomized algorithms with score at least $\gamma^n$ on every input of $f^n$, maximum cost $\R_{\gamma^n}(f^n)$.
The set $\calR_\mu$ is non-empty since by definition there is an algorithm that satisfies these two properties. 
And every leaf $\ell$ of the decision trees in the support of the randomized algorithms in this set have score $\sss_{f^n}^\mu(\ell) \ge (\frac12)^n > 0$ by the boundedness of scoring measures.

For any $R \in \calR_\mu$ and leaf $\ell$ of $R$,
\[
	\sss_{f^n}^\mu(\ell) \cdot \cost(\ell) 
 \le \sss_{f^n}^\mu(\ell) \cdot \R_{\gamma^n}(f^n).
\]
So
\if\conf1
\begin{align*}
\scost^\mu(R) 
&= \frac{\E_{\ell}[ \sss_{f^n}(\ell) \cdot \cost(\ell) ]}{\bars_{f^n}^\mu(R)} \\
&\le \frac{\E_{\ell}[ \sss_{f^n}(\ell) \cdot \R_{\gamma^n}(f^n) ]}{\bars_{f^n}^\mu(R)} \\
&= \R_{\gamma^n}(f^n)
\end{align*}
\else
\[
	\scost^\mu(R) 
	= \frac{\E_{\ell}[ \sss_{f^n}(\ell) \cdot \cost(\ell) ]}{\bars_{f^n}^\mu(R)} 
    \le \frac{\E_{\ell}[ \sss_{f^n}(\ell) \cdot \R_{\gamma^n}(f^n) ]}{\bars_{f^n}^\mu(R)} 
	= \R_{\gamma^n}(f^n)
\]
\fi
and $\sR_{\gamma^n}(f^n) = \max_\mu \sR_{\gamma^n}^\mu(f^n) \le \max_\mu \min_{R \in \calR_\mu} \scost^\mu(R) \le \R_{\gamma^n}(f^n)$.
\end{proof}

\subsection{Score and Discounted Score}

One of the basic properties of the discounted score measure that we will require in the proof of \cref{lem:maxDS-sR} is that for any score parameter $\gamma$ in the non-trivial range between the score of the guessing algorithm and $1$, there exists some discount rate $\alpha^*$ for which the maximum discounted score $\maxDS_{\alpha^*}^\mu(f)$ is achieved by a randomized algorithm $R$ with score exactly equal to $\gamma$.

\begin{proposition}
\label{prop:success-discount}
Fix any function $f$ and distribution $\mu$.
For every success parameter $\gamma$ that is at least as large as the success probability of the trivial guessing algorithm, there exists a discount rate $\alpha^*$ and a randomized algorithm $R^*$ that satisfies
\[
	\ds_{f,\alpha^*}^\mu(R^*) = \maxDS_{\alpha^*}^\mu(f)
	\qquad \mbox{and} \qquad
	\barsfmu(R^*) = \gamma.
\]
\end{proposition}

\begin{proof}
Note first that since the discounted score of a randomized decision tree $R$ is by definition
\[
	\ds_{f,\alpha}^\mu(R) = \E_{D \sim R} \E_{x \sim \mu} [ \sfmu(D(x)) \cdot e^{-\alpha \, \cost(D(x))} ],
\]
there is always a deterministic decision tree $D$ in the support of $R$ with discounted score at least as large as that of $R$.
As a result, the maximum discounted score $\maxDS_\alpha^\mu(f)$ is always attained by a deterministic decision tree.

For each deterministic decision tree $T$, define the function $\phi_T \colon \alpha \mapsto \ds_{f,\alpha}^\mu(T)$.
The function $\phi_T$ is continuous.
Then for any two trees $T_1$ and $T_2$, the function $\phi_{T_1} - \phi_{T_2}$ is also continuous.
As a result, the set $\{\alpha : \ds_{f,\alpha}^\mu(T_1) \ge \ds_{f,\alpha}^\mu(T_2)\}$, which is the preimage of $\phi_{T_1} - \phi_{T_2}$ on the closed set $[0,\infty)$, is a closed set.

Let us write $S_\ge, S_\le \subseteq [0,\infty)$ to represent the sets of discount rates for which the maximum discount score is achieved by a deterministic algorithm with score at least and at most $\gamma$, respectively. 
The set $S_\ge$ is closed, since we can write it as the finite union of finite intersections of closed sets
\[
	S_\ge = \bigcup_{T_1 : \barsfmu(T_1) \ge \gamma} \bigcap_{T_2} \{ \alpha : \ds_{f,\alpha}^\mu(T_1) \ge \ds_{f,\alpha}^\mu(T_2)\}.
\]
Similarly, $S_\le$ is also closed.
And these two closed sets partition their universe, i.e., $S_\ge \cup S_\le = [0,\infty)$.
As a result, their intersection $S_\ge \cap S_\le$ is not empty.
Let $\alpha^*$ be any element in this intersection.

By construction, for our choice of $\alpha^*$ there exist two deterministic algorithms $T_1$ and $T_2$ that satisfy
\[
	\ds_{f,\alpha^*}^\mu(T_1) = \ds_{f,\alpha^*}^\mu(T_2) = \maxDS_{\alpha^*}^\mu(f)
\]
and
\[
	\barsfmu(T_1) = \gamma - \delta_1, \quad \barsfmu(T_2) = \gamma + \delta_2 
\]
for some parameters $\delta_1,\delta_2 \ge 0$.
Define the randomized algorithm $A$ to be the convex combination
\[
	A = \frac{\delta_2}{\delta_1+\delta_2} T_1 + \frac{\delta_1}{\delta_1 + \delta_2} T_2.
\]
The discounted score of $A$ is
\[
	\ds_{f,\alpha^*}^\mu(A) = \maxDS_{\alpha^*}^\mu(f)
\]
and its score is
\if\conf1
\begin{align*}
    \barsfmu(A) 
 &= \frac{\delta_2}{\delta_1+\delta_2} \barsfmu(T_1) + \frac{\delta_1}{\delta_1+\delta_2} \barsfmu(T_2) \\
  &= \frac{\delta_2(\gamma - \delta_1) + \delta_1(\gamma+\delta_2)}{\delta_1+\delta_2} = \gamma. \qedhere
\end{align*}
\else
\[
	\barsfmu(A) 
 = \frac{\delta_2}{\delta_1+\delta_2} \barsfmu(T_1) + \frac{\delta_1}{\delta_1+\delta_2} \barsfmu(T_2)
  = \frac{\delta_2(\gamma - \delta_1) + \delta_1(\gamma+\delta_2)}{\delta_1+\delta_2} = \gamma. \qedhere
\]
\fi
\end{proof}